\documentclass[journal,11pt, onecolumn, draftcls]{IEEEtran}

\usepackage[cmex10]{amsmath}
\usepackage{cite,amssymb,amsthm,array, color,algorithm,algorithmic}

\usepackage[pdftex]{graphicx}
\graphicspath{{./figures/}{./}}
\DeclareGraphicsExtensions{.pdf,.jpeg,.png,.eps}

\usepackage{subcaption}

\IEEEoverridecommandlockouts

\theoremstyle{plain}
\newtheorem{lemma}{Lemma}
\newtheorem{theorem}{Theorem}

\theoremstyle{definition}

\theoremstyle{remark}

\title{Identification of Linear Time-Varying Systems Through Waveform Diversity}

\author{Andrew~Harms, \IEEEmembership{Member,~IEEE, }%
Waheed~U.~Bajwa, \IEEEmembership{Senior Member,~IEEE, }%
and Robert~Calderbank, \IEEEmembership{Fellow,~IEEE}%
\thanks{A.~Harms and R.~Calderbank are with the Department of Electrical and Computer Engineering, Duke University, Durham, NC 27708.}%
\thanks{W.~U.~Bajwa is with the Department of Electrical and Computer Engineering, Rutgers University, Piscataway, NJ 08854.}%
\thanks{A.~Harms and R.~Calderbank are supported in part by the Air Force Office of Scientific Research under grant FA9550-13-1-0076 from the Complex Networks program.}%
}

\begin{document}
\maketitle

\begin{abstract}
Linear, time-varying (LTV) systems composed of time shifts, frequency shifts, and complex amplitude scalings are operators that act on continuous finite-energy waveforms.  This paper presents a novel, resource-efficient method for identifying the parametric description of such systems, i.e., the time shifts, frequency shifts, and scalings, from the sampled response to linear frequency modulated (LFM) waveforms, with emphasis on the application to radar processing.  If the LTV operator is probed with a sufficiently diverse set of LFM waveforms, then the system can be identified with high accuracy.  In the case of noiseless measurements, the identification is perfect, while in the case of noisy measurements, the accuracy is inversely proportional to the noise level.  The use of parametric estimation techniques with recently proposed denoising algorithms allows the estimation of the parameters with high accuracy.

\end{abstract}

\section{Introduction}
In active sensing, a physical system is probed by a known waveform to identify the physical parameters of the system by processing the system response to the input waveform.  This paper is concerned with physical systems that are modeled by linear time-varying (LTV) systems described by time and frequency shifted versions of the input waveform.  For example, the response $y(t)$ to an input waveform $x(t)$ that is time shifted by $\tau$ and frequency shifted by $f$ would look like
\begin{equation}\label{eq:LTV-system-example}
  y(t) = x(t-\tau)e^{j2\pi f t}.
\end{equation}
The time-varying nature of the system is due to the modulation by the complex sinusoid, which is time-dependent.  Such LTV system models are important in radar processing, channel estimation for communication systems, and other areas, because the time and frequency shifts directly relate to path distances and velocities of physical objects that affect the signal returns.  In the communications literature \cite{kailath-TVCC-1962,bello-1962,bello-1963}, the time and frequency shifts are often modeled probabilistically, while in radar applications they are modeled as deterministic and fixed over short time intervals.  This paper will concentrate on the latter.

LTV systems are good models for multipath channels in which multiple copies of the transmitted signal, each with a time delay and a Doppler shift, are acquired by the receiver.  Assuming the bandwidth of the signal is small relative to the carrier frequency, the Doppler shift is adequately modeled by a frequency shift (see Section~\ref{sec:ltv-description} for a more details).  Fig.~\ref{fg:multi-path-scene} depicts a reference scene to illustrate a possible multipath channel.  Here, the transmitter and receiver are collocated, as is the case in a monostatic radar system.  Each object is identified by its path length, or, through the speed of light, the time delay from transmission to reception.  The moving car and airplane are further identified by a Doppler shift determined by their (radial) velocities.  The receiver processes the received signal using the transmitted signal as a reference.  The approach presented in this paper, as well as traditional approaches, is also readily adaptable to other scenarios, such as a non-collocated transmitter and receiver.  In this case, the receiver uses prior information about the transmitted signal, e.g., predefined pilot tones that are common in communication systems, or estimates the transmitted signal from the path directly from transmitter to receiver.  The method can also be adapted to a multi-antenna scenario in which the array is electronically steered.
For clarity of exposition, this paper concentrates on a collocated transmitter and receiver using a single omni-directional antenna for each, so that we may assume the receiver has a perfect copy of the transmitted signal.

\subsection{Existing Techniques for Identification}
Traditional processing employs a matched filter (MF), which correlates the received signal against hypothesized time and frequency shifts of the probe signal. Matched filtering is the maximum likelihood estimator for a single return (i.e., scatterer) in white Gaussian noise \cite{skolnik08}, but the detection of multiple targets is hampered by the spreading of the target peak.  The spreading is captured in the ambiguity function (i.e., 2-dimensional cross-correlation of time and frequency shifts) of the probe signal and limits the resolution of the time shifts and frequency shifts of targets in close proximity\cite{woodward-radar1953}.  An example ambiguity function is shown in Fig. \ref{fg:lfm-ambiguity-function} for a linear frequency modulated (LFM) pulse.  The waveform will be described in greater detail in Section \ref{sec:LFM-waveform}, but the detail to note is the line of large intensity from $-50$ ms to $50$ ms that couples the time and frequency shifts.  Thus, even in the absence of noise the detection of multiple returns is fundamentally limited by the waveform itself, which is captured by the ambiguity function.  Two targets described by parameters that happened to fall on that line would be indistinguishable from one another.  Matched filtering, or approximations to it, still finds large utility in practice because it is simply and efficiently computed by the fast Fourier transform \cite{skolnik08}.

\begin{figure}[tp]
  \centering
    \setlength{\unitlength}{.25in}
  \begin{picture}(12,6)
    \put(0.0, 0.5){\makebox{Tx/Rx}}
    \put(1.0, 1.0){\vector(3,1){8.9}}
    \put(9.8, 4.1){\vector(-3,-1){8.0}}
    
    \put(1.0, 1.0){\vector(1,2){2.0}}
    \put(2.8, 5.1){\vector(-1,-2){1.5}}
    
    \put(1.0, 1.0){\vector(1,0){3.0}}
    \put(4.0, 0.8){\vector(-1,0){2.5}}
    \put(2.6, 5.5){\makebox{Building}}
    \put(10.2, 4.5){\makebox{Airplane}}
    \put(4.4, 1.0){\makebox{Car}}
    \thinlines
    \qbezier(8.8, 6.0)(10.6, 4.0)(10.7,1.0)
    \qbezier(1.0, 6.0)(5.0, 5.0)( 6.0, 1.0)
    \qbezier(1.0, 4.2)(3.5, 3.5)( 4.2, 1.0)
  \end{picture}
   
  \caption{In a multipath channel, the received signal (Rx) contains multiple copies of the transmitted signal (Tx) that are reflections from various objects.  If the Tx and Rx are collocated, then the time shift is proportional to the distance from the Tx/Rx to the object and the frequency shift is proportional to the object's radial velocity.  Note that any objects located along each circle would impart the same time shift.}
  \label{fg:multi-path-scene}
\end{figure}
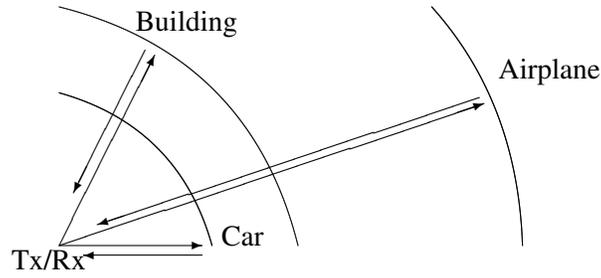

\begin{figure}[tp]
  \centering
  \includegraphics[width=0.85\columnwidth]{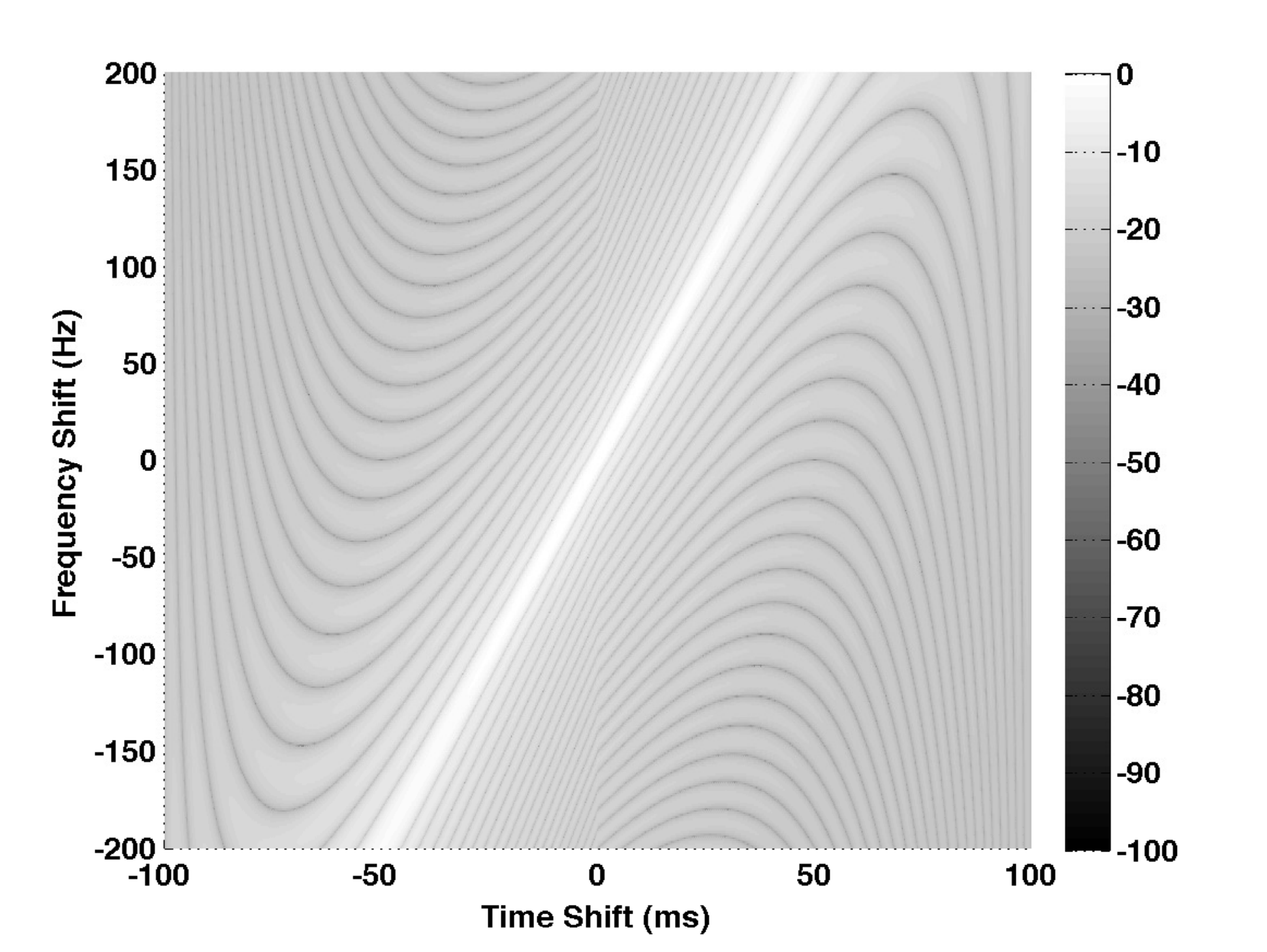}
  \caption{Ambiguity function of an LFM pulse.}
  \label{fg:lfm-ambiguity-function}
\end{figure}

Recently, other techniques have been proposed that are not fundamentally limited by the signal ambiguity function.  One technique uses ideas from the compressed sensing literature to efficiently identify LTV systems composed of a small number of returns, relative to the entire time-frequency shift space considered, and whose parameters live on a discretized grid of the time-frequency shift space \cite{strohmer-cs-radar-2009}.  This technique leverages the rich set of analysis techniques and recovery algorithms offered by compressed sensing.  The chief disadvantage is the necessity to discretize the time-frequency shift space because real-world LTV systems rarely conform to this assumption.  The mis-match to the assumed discretized basis has been shown to cause poor performance \cite{chi11}.  A similar but more general idea is explored in \cite{heckel-bolcskei-2013} in which fundamental limits of identification are related to the time-frequency shift \emph{spread} of the LTV operator using a similar discretization of the time-frequency shift space.

Parametric techniques have also been proposed as an alternative to matched-filter processing, including \cite{friedlander-parametric2012} and \cite{bajwa11-radar}.  The common thread to both of these approaches is a sequential recovery of the time shift followed by the frequency shift, or vice-versa.  A primary drawback is the sequential nature of the identification, especially when noise is considered.  Errors in the first step propagate to the second step and hamper the second stage of recovery.  In fact, the technique proposed in \cite{bajwa11-radar} does not consider noise in the analysis.  These techniques are also less efficient because the targets must be resolvable in both steps of the recovery.  The disadvantages will be discussed in more detail later and compared to the requirements of the proposed technique.

\subsection{Our Contributions}
We propose a technique that allows simultaneous identification of the time shifts, frequency shifts, and amplitude scalings of the returns from the LTV system.  The technique utilizes linear frequency modulation (LFM) pulses, i.e., linear chirps, as the probing waveform.  An LFM waveform has the form
$$x(t) = e^{j2\pi f_ct^2}g(t)$$
where $g(t) = 0$ $\forall t \not\in [0,T_p]$ and $f_c$ is the \emph{chirp rate} that defines how fast the frequency sweeps.  The returns from these probing waveforms, after pre-processing, are a superposition of sinusoids with frequency determined by the time and frequency shift of the target.  If a diverse set of LFM waveforms with different chirp rates $f_c$ is used to probe the system, then we can determine the time shifts and frequency shifts of the targets that produced the returns.  When multiple targets are present, we show that a diverse set of LFM pulses is sufficient to recover the description of each target.  In the case of noisy measurements, we show that the error in the estimated target parameters is proportional to the signal-to-noise ratio (SNR).  We also characterize the resource efficiency of the technique through analysis of the time-bandwidth product of the LFM waveform, which is characterized by the largest possible time and frequency shift of the target returns.  A comparison of the technique presented in this paper and selected other techniques is summarized in Table \ref{tab:summary-comparison}.  Notice that our approach improves on both the resolution and the number of samples required.  This work builds upon initial work \cite{harms-camsap-2013} in which we presented preliminary analysis of using a diverse set of LFM pulses for recovery from noiseless measurements.  In the present work, we expand the analysis to noisy measurements, expand the discussion of the usefulness of diversity, and include expanded Monte Carlo numerical experiments.

\begin{table*}
  \centering
  \begin{tabular}{p{1.2in}  c  c  c  c  c}
    Approach & Num. of Samples & Time-Bandwidth & Resolution w/o noise & Avg. Complexity & Resolution w/ noise \\
    \hline
    Matched Filter \cite{skolnik08} & $W\cdot \tau_{\max}$ & $\frac{1}{\Delta\tau \cdotp \Delta f}$ & $\Delta\tau \propto \frac{1}{W}$, $\Delta f \propto \frac{1}{M}$ & poly($\Delta\tau,\Delta f$) & $\Delta\tau \propto \frac{1}{W}$, $\Delta f \propto \frac{1}{M}$ \\
    \hline
    Friedlander \cite{friedlander-parametric2012} & $M\cdot K^2$ & $K^2$ & $\infty$ & poly($K$) & $\Delta\tau \propto \frac{1}{M^3N}$, $\Delta f \propto \frac{1}{N^3M}$ \\
    \hline
    Bajwa et al. \cite{bajwa11-radar} & -- & $K^2$ & $\infty$ & poly($K$) & -- \\
    \hline
    Herman, Strohmer \cite{strohmer-cs-radar-2009} & $M\cdot K^2$ & $K^2$ & $\Delta\tau \propto \frac{1}{W}$, $\Delta f \propto \frac{1}{M}$ & poly($K$) & $\Delta\tau \propto \frac{1}{W}$, $\Delta f \propto \frac{1}{M}$ \\
    \hline
    Harms et al. (this paper) & $M\cdot K$ & $K^2$ & $\infty$ & poly($K$) & $\Delta\tau, \Delta f \propto \frac{1}{(MN)^3}$ \\
    \hline
  \end{tabular}
  \caption{Summary comparison of the approach described in this paper to several relevant alternative approaches.  The values listed indicate how the quantity scales in terms of the number of samples $K$, number of pulses $M$, number of samples per pulse $N$, and maximum time shift considered $\tau_{\max}$.  Please see Sections \ref{sec:ltv-description} and \ref{sec:LTV-system-ID} for more details about these quantities and Section \ref{sec:noiseless-discussion} for more details about the comparison.  Other constants are left off for clarity, and $\text{poly}(\cdot)$ means polynomial scaling.  A `--' indicates that the column is not addressed by the approach.}
  \label{tab:summary-comparison}
\end{table*}

The remainder of this paper is organized as follows.  Section \ref{sec:ltv-description} sets up the model for LTV systems and specifies the response of these systems to a train of LFM pulses.  Section \ref{sec:LTV-system-ID} discusses the proposed processing scheme that first uses analog preprocessing on the received LFM pulses to setup a digital frequency recovery algorithm.  Section \ref{sec:noiseless-recovery} analyzes the proposed recovery algorithm performance from the noiseless LTV system response, while Section \ref{sec:noisy-estimation} analyzes the performance using a noisy system response.  We finish up with numerical experiments to verify the proposed procedure in Section \ref{sec:numerical-experiements} and conclude in Section \ref{sec:conclusion}.

\section{LTV System Description and Response}\label{sec:ltv-description}
An LTV system is an operator described by time shifts, frequency shifts, and complex scalings.  The LTV operator acts on \emph{probing waveforms} and produces a response that consists of time-shifted, frequency-shifted, and scaled copies of the probing waveform.  Each parameter triplet $(\tau_k, f_k, c_k)$ is associated with a \emph{target} or \emph{object}, owing to the physical interpretation of a radar scene or multipath communication system.  The objective is to identify the LTV operator by estimating each triplet $(\tau_k, f_k, c_k)$.  In causal systems, as considered in this paper, the time shifts are time delays, i.e., $\tau_k > 0$.  The system response is a superposition of the modified probing waveforms.  Fig. \ref{fg:block-diagram} shows a block diagram of the LTV system description where $x(t)$ denotes the probing waveform that illuminates the system and
$h(\tau,f,c) = c\cdotp x(t-\tau)e^{j2\pi ft}$
is the operator corresponding to a single target.  Assuming there are $K$ targets, the received signal\footnote{The received and transmitted signals are modeled as complex signals in this paper.  This would be implemented in practice using I/Q modulation \cite{dsp-proakis}.} is
\begin{equation}\label{eq:LTV-system}
  y(t) = \sum_{k=1}^{K}c_k x(t-\tau_k)e^{j2\pi f_k t} + \varepsilon(t)
\end{equation}
for $t \in [0,\mathcal{T}]$ where $\mathcal{T}$ is the processing interval (during which the parameters are assumed fixed) and $\varepsilon(t)$ is a noise term.
The time and frequency shifts are assumed to be limited\footnote{In practice, these limitations are set by the physical limits of the system and scenario.  For example, a radar system has limited sensitivity and can only detect returns from targets over some finite range and targets are limited to some finite velocity.}, i.e., $f_k \in (-f_{\max}, f_{\max})$ and $\tau_k \in (0,\tau_{\max})$ $\forall k$ such that $\max(\tau_k) < \tau_{\max}$ and $\max(|f_k|) < f_{\max}$.  The choice of $\mathcal{T}$ is important to ensure the parameters are (approximately) fixed over the processing interval.
A moving target that has a frequency shift due to the Doppler effect changes range, and hence the time shift changes, over time, so $\mathcal{T}$ must be small enough so that the change in time shift is negligible.

As a quick aside about an LTV operator that produces a response of the form \ref{eq:LTV-system}, we note that general LTV operators with a continuous spreading function, such as those considered in \cite{heckel-bolcskei-2013}, can be decomposed into a finite sum of discrete targets of the form \ref{eq:LTV-system} if the spreading function admits a Fourier decomposition.  We also note that the advantages of the method described in this paper are found when such a decomposition results in small $K$ because our identification results are a function of $K$.  Larger $K$ will require more resources.  However, in many applications $K$ is small or can be well-approximately by a small number of targets.

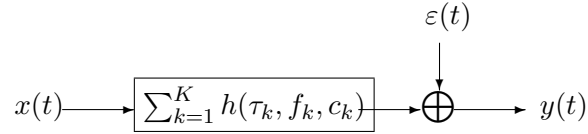
\begin{figure}[tp]
  \centering
    \setlength{\unitlength}{.25in}
  \begin{picture}(12,3)
    \put(0.0, 1.0){\makebox{$x(t)$}}
    \put(1.0, 1.1){\vector(1,0){1.5}}
    \put(2.5, 1.0){\framebox{$\sum_{k=1}^Kh({\tau_k,f_k,c_k})$}}
    \put(7.2, 1.1){\vector(1,0){1.3}}
    \put(8.5, 1.0){\makebox{\large$\bigoplus$}}
    \put(8.9, 2.5){\vector(0,-1){1.0}}
    \put(8.6, 2.9){\makebox{$\varepsilon(t)$}}
    \put(9.2, 1.1){\vector(1,0){1.5}}
    \put(11.0, 1.0){\makebox{$y(t)$}}
  \end{picture}
  \caption{The LTV system described by \eqref{eq:LTV-system}.  The system introduces time shifts $\tau_k$, frequency shifts $f_k$, and complex scalings $c_k$ that modify the probing waveform $x(t)$.  Additive noise $\varepsilon(t)$ corrupts the signal returns to produce the received signal $y(t)$.}
  \label{fg:block-diagram}
  
\end{figure}

\subsection{Probing Waveform: Linear Frequency Modulated Pulses}\label{sec:LFM-waveform}
The probing waveform must be designed to produce an identifiable system response.  The important property of a probing waveform is that it provides a sufficient number of degrees of freedom in the system response.  For example, a pure tone does not provide any information about the time shift because the phase of the return is corrupted by the unknown phase imparted by the target.

The waveforms considered are linear frequency modulated (LFM) pulses, or windowed \emph{chirps}, which enjoy wide use in radar applications \cite{skolnik08}.  These waveforms are complex sinusoids in which the frequency sweeps linearly in time across some bandwidth.
Consider a train of $M$ such pulses,
\begin{align*}
  x(t) &= \sum_{m=0}^{M-1} x_m(t-mT)
\end{align*}
where each pulse is a windowed LFM waveform with sweep rate $f_c^m$ and frequency offset $f_0^m$
\begin{align}
  x_m(t) &= \mathrm{e}^{j2\pi ( f_c^{m} t^2 + f_0^{m} t) }g(t) \label{eq:lfm-pulse-m}
\end{align}
with a square window function
\begin{equation*}
  g(t) = \begin{cases} 1,& 0 \leq t \leq T_p \\ 0,& \text{otherwise.} \end{cases}
\end{equation*}
Square windows allow a clear analysis, though other window functions are possible, e.g., a continuous window that tapers at each end, but beyond the scope of the current analysis.  The pulse duration is $T_p$, and $T$ is the pulse repetition interval (PRI).  The time-frequency plot of an example pulse at baseband
is shown in Fig. \ref{fig:lfm-waveform}.  The time-frequency characterization, however, does not completely capture the spectral content of the LFM pulse because the total occupied bandwidth is slightly larger than the difference between the starting and ending frequencies due to the windowing in time of the pulse.  The Fourier transform of \eqref{eq:lfm-pulse-m} is
\begin{equation}\label{eq:LFM-Fourier}
  X(f) = \mathcal{F}\{x_m(t)\} = \int_{0}^{T_p}\mathrm{e}^{j2\pi(f_c^mt^2 + (f_0^m - f)t)}dt.
\end{equation}
The integral in \eqref{eq:LFM-Fourier} is difficult to evaluate due to the quadratic term in the exponent.  The integral can be formulated in terms of Fresnel integrals and numerically evaluated as in \cite{klauder-bstj-1960}.  The upshot is that most of the energy is contained in a bandwidth of approximately $2f_c^mT_p$ with some energy leaking outside of this bandwidth.  Because the pulse duration is $T_p$, the time-bandwidth product of a single pulse is $2f_c^mT_p^2$.
Other chirp, or chirp-like, pulses are possible such as a linear frequency-stepped pulse as analyzed in \cite{friedlander-parametric2012}.  In this case, $f_c^{m} = 0$ and $f_0^{m} = f_0 + \delta f\cdotp m$ where $\delta f$ is the frequency step for each pulse.

\begin{figure}[tp]
  \centering
  \begin{subfigure}{0.49\columnwidth}
    \centering
    \includegraphics[width=0.99\columnwidth]{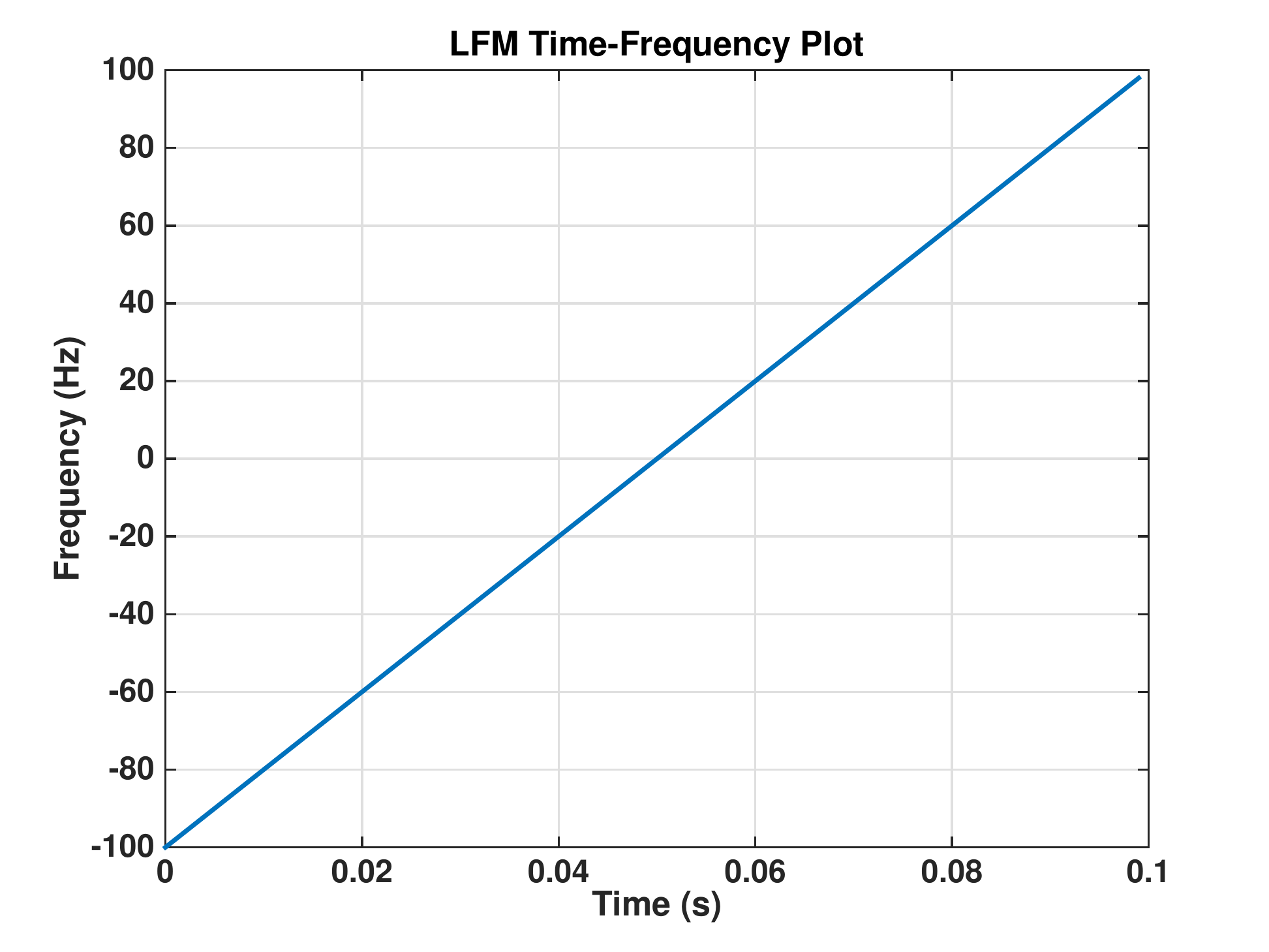}
    \caption{Time-frequency plot.}
    \label{fig:lfm-waveform}
  \end{subfigure}
  \begin{subfigure}{0.49\columnwidth}
    \centering
    \includegraphics[width=0.99\columnwidth]{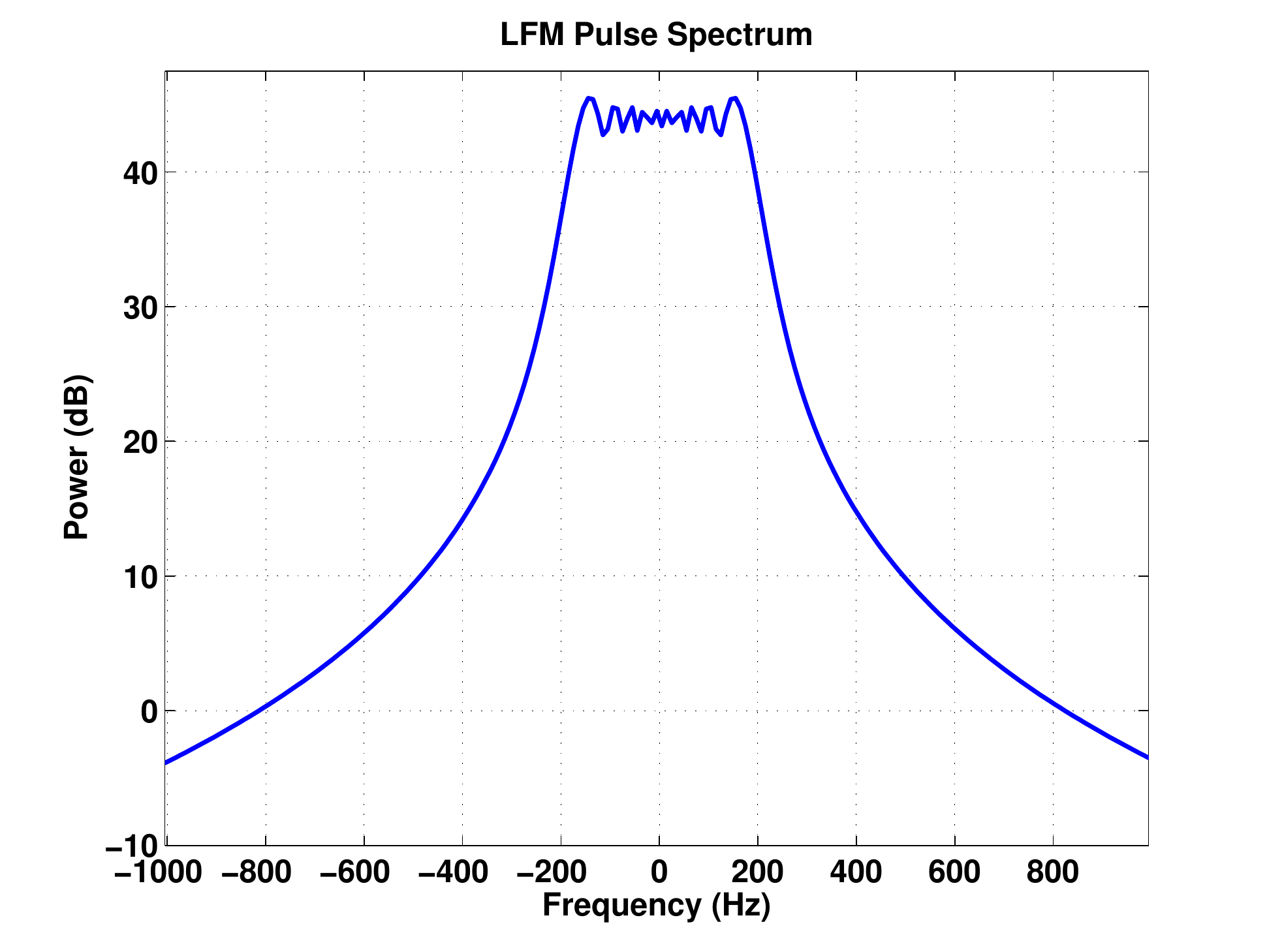}
    \caption{Fourier transform.}
    \label{fg:lfm-spectrum}
  \end{subfigure}
  \caption{The time-frequency plot and the Fourier transform characterize the spectral content of an LFM pulse with $f_0 = -100$ Hz, $f_c = 2000$ Hz/s, and $T_p = 0.1$ s. Most energy is contained in a bandwidth of $\approx 2f_cT_p$}
  \label{fg:lfm-frequency-characterization}
\end{figure}

\subsection{LTV System Response to LFM Pulses}
The response of the LTV operator to these LFM pulses is
\begin{align*}
  y(t) &= \sum_{k=1}^{K}\sum_{m=0}^{M-1}c_k x_m(t-\tau_k-mT)\mathrm{e}^{j2\pi f_k t} + \varepsilon(t) = \sum_{m=0}^{M-1}y_m(t)
\end{align*}
where
\begin{equation}\label{eq:received-pulse-m}
  y_m(t) = \sum_{k=1}^{K}c_k x_m(t-\tau_k-mT)\mathrm{e}^{j2\pi f_k t} + \varepsilon_m(t)
\end{equation}
is the received signal for the $m^{th}$ pulse with 
\begin{equation*}
   \varepsilon_m(t) = \begin{cases}  \varepsilon(t),& mT \leq t \leq (m+1)T \\ 0,& \text{otherwise} \end{cases}
\end{equation*}
the windowed noise process over the time interval of the $m^{th}$ pulse.
The window $g(t)$ limits the temporal extent of $x_m(t)$, so $y_m(t)$ is guaranteed to be zero outside of the interval
$mT \leq t \leq T_p+\tau_{\max}+mT$
for $m = 0,\ldots,M-1$.
If $T \geq T_p + \tau_{\max}$, then the received pulses $y_{m_1}(t)$ and $y_{m_2}(t)$ do not overlap for $m_1 \neq m_2$.  We therefore have a \emph{guard interval} of $T_g = T - T_p$ during which no signal is transmitted, and the requirement for non-overlapping received pulses is $T_g \geq \tau_{\max}$.  We also restrict the \emph{measurement interval} for the $m^{th}$ pulse to
$$mT + \tau_{\max} \leq t \leq mT + T_p$$
to ensure all returns from the $m^{th}$ pulse, and only those from the $m^{th}$ pulse, are present.  We denote the measurement interval as $T_o = T_p - \tau_{\max}$.  Fig.~\ref{fg:timing-diagram} provides a pictorial summary of the timing requirements.

\begin{figure}[tp]
  \centering
    \setlength{\unitlength}{.25in}
  \begin{picture}(12,2)
    \put(0.0, 1.0){\vector(1,0){11.0}}
    \put(11.2, 0.8){\makebox{$t$}}
    \put(0.0, 0.8){\line(0,1){0.4}}
    \put(-0.1, 0.2){\makebox{$0$}}
    \put(2.0, 0.8){\line(0,1){0.4}}
    \put(1.8, 0.2){\makebox{$\tau_{\max}$}}
    \put(6.0, 0.8){\line(0,1){0.4}}
    \put(5.8, 0.2){\makebox{$T_p$}}
    \put(10.0, 0.8){\line(0,1){0.4}}
    \put(9.8, 0.2){\makebox{$T$}}
    \put(3.7, 1.3){\makebox{$T_o$}}
    \put(3.5, 1.5){\vector(-1,0){1.2}}
    \put(4.5, 1.5){\vector( 1,0){1.2}}
    \put(7.8, 1.3){\makebox{$T_g$}}
    \put(7.6, 1.5){\vector(-1,0){1.2}}
    \put(8.5, 1.5){\vector( 1,0){1.2}}
  \end{picture}
  \caption{Timing diagram of a pulse, the measurement period, and the guard period.  We require that $T_g = T-T_p \geq \tau_{\max}$ and $T_p >\tau_{\max}$.  The measurement period $T_o$ has a lower bound (described later) that depends on the number of targets and the noise power.}
  \label{fg:timing-diagram}
  
\end{figure}
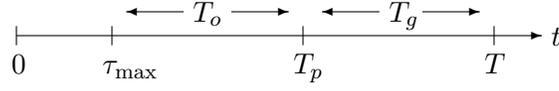

The received pulse \eqref{eq:received-pulse-m} is a superposition of frequency-offset chirps.  Substituting \eqref{eq:lfm-pulse-m} into \eqref{eq:received-pulse-m} yields
\begin{equation*}
  y_m(t) = \sum_{k=1}^{K}c_k \mathrm{e}^{j2\pi \theta_k^{m}}\mathrm{e}^{j2\pi \nu_k^{m}t} \mathrm{e}^{j2\pi(f_c^{m}t^2 + f_0^{m} t) }g(t-\tau_k) + \varepsilon_m(t)
\end{equation*}
where
\begin{align}
  \theta_k^{m} &= f_c^{m}\tau_k^2 \label{eq:phase-mapping-positive} \\
  \nu_k^{m} &= f_k - 2f_c^{m}\tau_k \label{eq:freq-mapping-positive}
\end{align}
are, respectively, the phase offset and frequency offset of the chirp associated with the $k^{th}$ target return and determined by the time shift and frequency shift of the $k^{th}$ target.\footnote{Note that there is a phase term $f_0^m\tau_k$ that has been subsumed into the complex scaling $c_k$ assuming that $f_0^m = f_{RF}$, $\forall m$, is the RF center frequency common to every LFM pulse.  We present the analysis at baseband while acknowledging that $c_k$ has a component due to  $f_{RF}$.}
The \emph{pure chirp} term, $\mathrm{e}^{j2\pi(f_c^mt^2 + f_0^m t)}$, does not depend on any parameter of the $k^{th}$ target.

\section{Identification of LTV Systems}\label{sec:LTV-system-ID}
The goal in LTV system identification is to identify (i.e., estimate or recover) the composite time shifts, frequency shifts, and amplitude scalings by probing the system with known waveforms that provide sufficient diversity in the system response.  In the case of LFM probing waveforms, the identification is split into three parts: 1) analog preprocessing and sampling of the received signal, 2) frequency estimation (the parameters in \eqref{eq:phase-mapping-positive} and \eqref{eq:freq-mapping-positive}), and 3) matching of the recovered frequencies to determine the time and frequency shifts (see Fig. \ref{fg:processing-block-diagram}).

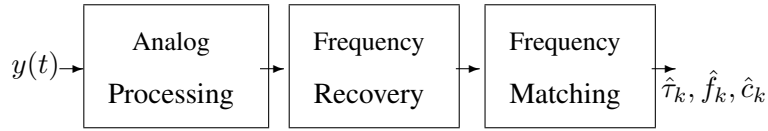
\begin{figure}[tp]
  \centering
    \setlength{\unitlength}{.25in}
  \begin{picture}(14,2.5)
    \put(0.0, 1.0){\makebox{$y(t)$}}
    \put(1.0, 1.1){\vector(1,0){0.5}}
    \put(1.5, 1.0){\framebox{\begin{tabular}{c}\small Analog \\ Processing\end{tabular} }}
    \put(5.2, 1.1){\vector(1,0){0.5}}
    \put(5.8, 1.0){\framebox{\begin{tabular}{c}\small Frequency \\ Recovery\end{tabular} }}
    \put(9.3, 1.1){\vector(1,0){0.5}}
    \put(9.9, 1.0){\framebox{\begin{tabular}{c}\small Frequency \\ Matching\end{tabular} }}
    \put(13.4, 1.1){\vector(1,0){0.5}}
    \put(13.6, 0.5){\makebox{$\hat{\tau}_k, \hat{f}_k, \hat{c}_k$}}
  \end{picture}
  \caption{Block diagram describing the identification of an LTV system.  The received signal $y(t)$ is dechirped and sampled.  The samples are used to recover the constituent sinusoids in the processed returns.  Finally, the recovered frequencies are matched to recover the LTV system description.}
  \label{fg:processing-block-diagram}
  
\end{figure}

\subsection{Analog Receiver Processing}
The analog preprocessing of the received returns is shown in Fig.~\ref{fg:receiver-diagram}.
The received signal is first \emph{dechirped} to remove the pure chirp component.  The dechirped received signal for the $m^{th}$ pulse is
\begin{align}\label{eq:preprocessing}
  \tilde{y}_m(t) &= \mathrm{e}^{-j2\pi (f_c^{m} t^2 + f_0^{m} t)}y_m(t) \\
  &= \sum_{k=1}^{K}c_k \mathrm{e}^{j2\pi \theta_k^{m}} e^{j2\pi \nu_k^{m} t}g(t-\tau_k) + \tilde{\varepsilon}_m(t) \notag
\end{align}
where $\tilde{\varepsilon}_m(t) = \mathrm{e}^{-j2\pi (f_c^{m} t^2 + f_0^{m} t)}\varepsilon_m(t)$.  The dechirped signal is a sum of complex sinusoids with frequencies $\nu_k^{m}$ and phases $\theta_k^m$.  The noise term $\tilde{\varepsilon}_m(t)$ is phase modulated by the dechirping process, but its magnitude is unaffected.  For example, if the noise process $\varepsilon_m(t)$ is independent complex (circularly symmetric) Gaussian noise, then the dechirped noise process has the same statistics.  The magnitude $|\tilde{\varepsilon}_m(t)|$ is unaffected and the phase remains uniformly distributed in $[0, 2\pi)$.

The dechirped signal is then sampled over the measurement interval so that the output measurements are
\begin{equation}\label{eq:measurement-samples}
  \tilde{y}_m[n] = \tilde{y}_m(nT_s) = \sum_{k=1}^{K}c_k \mathrm{e}^{j2\pi \theta_k^{m}} \mathrm{e}^{j2\pi \nu_k^{m} nT_s} + \tilde{\varepsilon}_m[n]
\end{equation}
for $n = 0,...,N-1$ with $T_s$ the sampling period and $\tilde{\varepsilon}_m[n] = \tilde{\varepsilon}_m(nT_s)$.  We first note that the Nyquist criterion (to prevent aliasing of the different sinusoids in \eqref{eq:measurement-samples}) requires that $T_s\cdotp \max(\nu_k^{m}) \leq \frac{1}{2}$, or
\begin{equation}\label{eq:sampling-rate}
  f_s \geq 2(f_{\max} + 2f_c^{m}\tau_{\max})
\end{equation}
where $f_s = \frac{1}{T_s}$ is the sampling rate.  Note that this Nyquist criterion does not depend directly on the bandwidth of the probing LFM waveform, which is approximately $2f_c^mT_p$.  Generally, $T_p \gg \tau_{\max}$, so the sampling constraint on $f_s$ is much smaller than the bandwidth of the LFM waveform.  There is also an implicit relation between $f_{\max}$ and $T_p$ requiring that the product cannot be too large, i.e., $f_{\max}\cdotp T_p < \eta$ for some constant $\eta$.  The product is proportional to the distance traveled by a moving target during time $T_p$ and cannot be too large to satisfy the assumption of the parameters remaining fixed.

The dechirping \eqref{eq:preprocessing} converts the time and frequency shifted LFM pulses into complex sinusoids with frequency and phase determined by the time and frequency shift.
We can write $c_k = |c_k|e^{j2\pi\phi_k}$ and \eqref{eq:preprocessing} becomes
\begin{align}\label{eq:positive-pulse-measurements}
  \tilde{y}_m[n] &= \sum_{k=1}^{K}|c_k| \mathrm{e}^{j2\pi\psi_k^m} \mathrm{e}^{j2\pi \nu_k^m nT_s} + \tilde{\varepsilon}_m[n]
\end{align}
with $\psi_k^m = \phi_k + \theta_k^m = \phi_k + f_c^{m}\tau_k^2$.  We have transformed a chirp estimation problem into a sinusoid estimation problem in which the frequencies and phases of the sinusoids are parametrically defined by the time shifts, frequency shifts, and phase offsets of the LTV operator.  The problem remains of recovering the parameters of the sinusoids and solving for the LTV time shifts and frequency shifts.

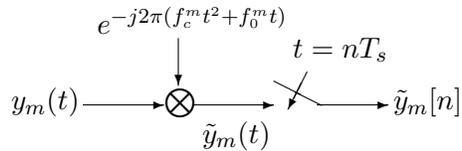
\begin{figure}[tp]
  \centering
    \setlength{\unitlength}{.25in}
  \begin{picture}(9.5,3.0)
    \put(0.0, 0.7){\makebox{$y_m(t)$}}
    \put(1.5, 0.8){\vector(1,0){1.7}}
    \put(3.2, 0.7){\makebox{$\bigotimes$}}
    \put(3.5, 2.2){\vector(0,-1){1.0}}
    \put(1.8, 2.3){\makebox{$e^{-j2\pi(f_c^m t^2 + f_0^m t)}$}}
    \put(3.8, 0.8){\vector(1,0){1.7}}
    
    \put(6.5, 0.8){\line(-2,1){1}}
    \put(6.2, 1.5){\vector(-1,-2){0.4}}
    \put(6.4, 0.8){\vector(1,0){1.5}}
    \put(5.9, 1.8){\makebox{$t = nT_s$}}
    
    \put(4.0, 0.0){\makebox{$\tilde{y}_m(t)$}}
    \put(8.0, 0.7){\makebox{$\tilde{y}_m[n]$}}
  \end{picture}
  \caption{The receiver analog processing chain.  The received signal is \emph{dechirped}, then sampled and sent to the digital processing stage.}
  \label{fg:receiver-diagram}
\end{figure}

\subsection{Digital Receiver Processing}
The dechirped and sampled signal $\tilde{y}_m[n]$ is the input to a digital processing stage that ultimately recovers the LTV system description by recovering the $K$ target triplets $(\tau_k, f_k, c_k)$.  The digital recovery proceeds in two steps (see Fig. \ref{fg:processing-block-diagram}).  The first step is recovery of the frequencies \eqref{eq:freq-mapping-positive} and phases \eqref{eq:phase-mapping-positive} from each pulse $m=1,\ldots,M$.  The parametric relationship between these frequencies and phases and the time shifts and frequency shifts is then exploited to recover the target triplets, as described in the following sections.  The case of noise free samples is analyzed first, followed by the case of noise corrupted samples.  The procedure is more clearly explained without considering noise in Section~\ref{sec:noiseless-recovery}, and the extension to the case of noisy measurements is straightforward in Section~\ref{sec:noisy-estimation}.

\section{Identification of LTV Systems from Uncorrupted Samples}\label{sec:noiseless-recovery}
We begin by analyzing the noise-free case with ${\varepsilon}(t) = 0$, and by extension $\tilde{\varepsilon}_m[n] = 0$.  The noise-free case is useful for two reasons.  First, it gives a baseline for performance in terms of several benchmarks considered.  Second, it provides an intuitive understanding of the identification procedure.  Given the filtered, dechirped, and sampled measurements \eqref{eq:measurement-samples}, we first recover the frequencies and phases of the resulting sinusoids and then extract the time and frequency shifts from these frequencies and phases.

\subsection{Recovery of Time Shifts and Frequency Shifts}
The first stage of the digital processing is recovery of the frequencies and phases from the noise-free samples \eqref{eq:measurement-samples}, i.e., with $\tilde{\varepsilon}[n] = 0$.  Specifically, $\psi_k^m$ and $\nu_k^m$ are recovered from \eqref{eq:positive-pulse-measurements}.
The samples $\tilde{y}_m[n]$ from each pulse are the input, and some frequency recovery algorithm is used to recover the constituent frequencies.  For completeness, we summarize the Kumaresan-Tufts (KT) algorithm \cite{tufts-kumaresan-1982} as one possible frequency recovery technique but emphasize that other parametric techniques, such as MUSIC or ESPRIT, could be easily used in its place.  Additionally, we could use non-parametric Fourier-based techniques.

The KT algorithm solves the prediction equation
$\mathbf{y} + \mathbf{Y}\mathbf{h} = \mathbf{0}$
where $\mathbf{h} = [h[1], \ldots, h[L]]^T$ are the coefficients of the predictor filter
\begin{equation*}
  H(z) = z^L + h[1]z^{L-1} + \cdots + h[L-1]z + h[L] = \prod_{i=1}^L(z-\hat{z}_i),
\end{equation*}
$L$ is the \emph{predictor order} of the filter, and $\hat{z}_i$ are the roots of the polynomial.
Using the forward-backward predictor matrix
\begin{equation*}
    \mathbf{Y} =
    \begin{bmatrix}
      y[L]     & y[L-1] & \cdots & y[1] \\
      \vdots & \vdots & \vdots & \vdots \\
      y[N-1] & y[N-2] & \cdots & y[N-L] \\
      y^*[2]  & y^*[3]  & \cdots & y^*[L+1] \\
      \vdots & \vdots & \vdots  & \vdots \\
      y^*[N-L+1] & y^*[N-L+2]  & \cdots & y^*[N]
     \end{bmatrix},
\end{equation*}
and forward-backward predictor vector $\mathbf{y}$
$$\mathbf{y} = \begin{bmatrix} y[L+1] & \cdots & y[N] & y^*[1] & \cdots & y^*[N-L] \end{bmatrix}^T,$$
the prediction equation solution is
$$\mathbf{h} = -(\mathbf{Y}^H\mathbf{Y})^{-1}\mathbf{Y}^H\mathbf{y} = -\mathbf{R}^{-1}\mathbf{r}$$
where
$\mathbf{R} = \mathbf{Y}^H\mathbf{Y}$
is the data correlation matrix and
$\mathbf{r} = \mathbf{Y}^H\mathbf{y}$
is the data correlation vector.  Notice that the structure of $\mathbf{Y}$ (and by extension $\mathbf{R}$) means they have rank $K$, the number of constituent sinusoids.

The prediction filter $H(z)$ has $L$ roots, denoted by $\hat{z}_i$, where $K$ roots lie on the unit circle and the rest reside inside the unit circle.  The recovered frequencies $\hat{\nu}_k$ are found from the $K$ roots on the unit circle
$$\hat{z}_k = e^{j2\pi\hat{\nu}_kT_s}.$$

The phases $\hat{\psi}_k^m$ and amplitudes $|\hat{c}_k|$ are recovered with the following least-squares problem.  Let $\mathrm{F}_{\hat{\nu}}$ be the matrix of sinusoids with the estimated frequencies $\hat{\nu}_k^m$ defined by
$[\mathrm{F}_{\hat{\nu}}]_{n,k} = e^{j2\pi\hat{\nu}_k^m n T_s}.$
Collecting the samples into a vector $\tilde{y}^m = [\tilde{y}^m[0],...,\tilde{y}^m[N-1]]^T$, the least-squares solution is
\begin{equation}\label{eq:least-squares-amplitude}
  \hat{\zeta} = \arg\min_{\zeta} ||\mathrm{F}_{\hat{\nu}}\zeta - \tilde{y}^m||_2^2
\end{equation}
with entries
$\hat{\zeta}_k = |\hat{c}_k|e^{j2\pi\hat{\psi}_k^m}.$
The procedure is summarized in Algorithm \ref{alg:frequency-phase-recovery}.

\begin{algorithm}
  \caption{KT algorithm for recovering the frequency and phase of complex sinusoids from sampled measurements}
  \label{alg:frequency-phase-recovery}
  \begin{algorithmic}[1]
    \STATE Data: $\tilde{y}_m[n]$
    \STATE Calculate coefficients $\mathbf{h} = -(\mathbf{Y}^H\mathbf{Y})^{-1}\mathbf{Y}^H\mathbf{y} = -\mathbf{R}^{-1}\mathbf{r}$.
    \STATE Find the roots on the unit circle, $|\hat{z}_k| = 1$, of $H(z)$.
    \STATE Calculate the frequencies $\hat{\nu}^m_k = \frac{T_s}{2\pi}\text{phase}(\hat{z}_k)$.
    \STATE Calculate $\hat{\psi}^m_k$ and $|\hat{c}_k|$ from the least-squares solution $\hat{\zeta}_k$.
  \end{algorithmic}
\end{algorithm}

Recovering $\psi_k^m$ and $\nu_k^m$ from a single LFM pulse (i.e., $M = 1$) is insufficient to determine the time and frequency shifts.  Recovery of $\nu_k^1$ provides a linear relationship \eqref{eq:freq-mapping-positive} to possible time shifts and frequency shifts that could produce a sinusoid of that frequency.  The recovered phase $\psi_k^1$ does not provide information about the time shift due to the unknown target phase $\phi_k$.  For example, the line with positive slope in Fig.~\ref{fg:lfm-single-target-constraints} shows the constraint imposed by \eqref{eq:freq-mapping-positive} in the case of a single target return ($K=1$).  There exist infinitely many solutions.

\begin{figure}[tp]
  \centering
    \includegraphics[width=0.50\columnwidth]{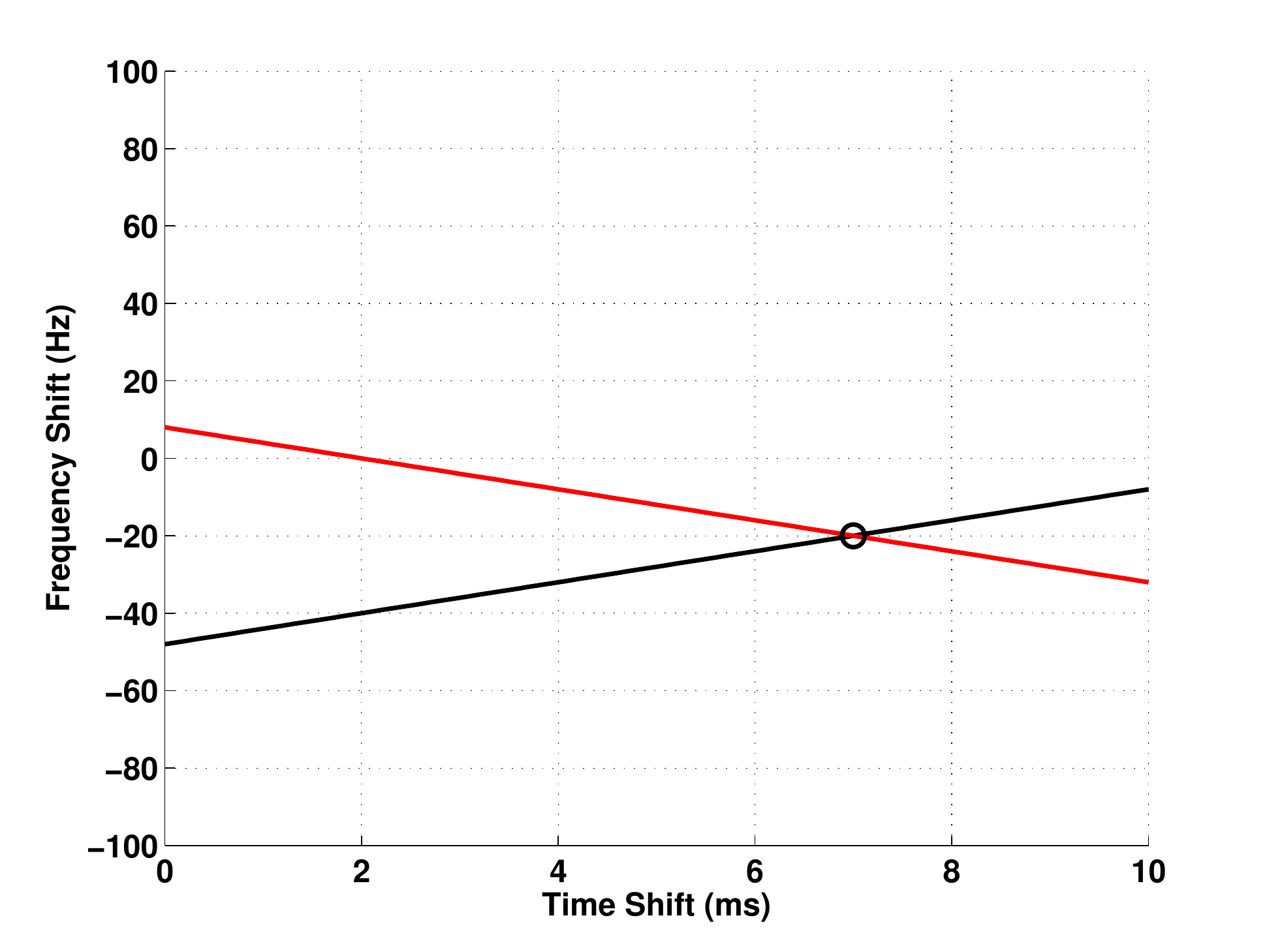}
  \caption{The constraints imposed by two LFM pulses probing an LTV system consisting of a single target with $\tau_{\max} = 0.01$s and $f_{\max} = 100$Hz.  The line with positive slope corresponds to the constraint imposed by the positive chirp, while the line with negative slope corresponds to the negative chirp.  The intersection of the lines is the time--frequency shift pair that explains the recovered frequency from each pulse.
  }
  \label{fg:lfm-single-target-constraints}
\end{figure}

However, we can send another LFM pulse with a different chirp rate, e.g., $f_c^2 = -f_c^1$.
The so-called \emph{positive} and \emph{negative} chirp in combination provide sufficient information to recover the time and frequency shift of the lone target.
The frequencies in this case are
\begin{align}
  \nu^{1}_k &= f_k - 2f_c^1\tau_k \ \ \ \text{and} \ \ \ \nu^{2}_k = f_k + 2f_c^1\tau_k\label{eq:pos-neg-freq-phase} 
\end{align}
and only one time--frequency shift pair satisfies both constraints simultaneously, shown by the intersection of the lines in Fig.~\ref{fg:lfm-single-target-constraints}.  The time shift and frequency shift are
$$\tau_1 = \frac{\nu_1^1 - \nu_1^2}{-4f_c^1} \ \ \ \text{and} \ \ \ f_1 = \frac{\nu_1^1 + \nu_1^2}{2}.$$
Note that the only requirement is $f_c^2 \neq f_c^1$ (for a single target in the absence of noise).

\begin{figure}[tp]
  \centering
  \begin{subfigure}{0.49\columnwidth}
    \centering
    \includegraphics[width=0.99\columnwidth]{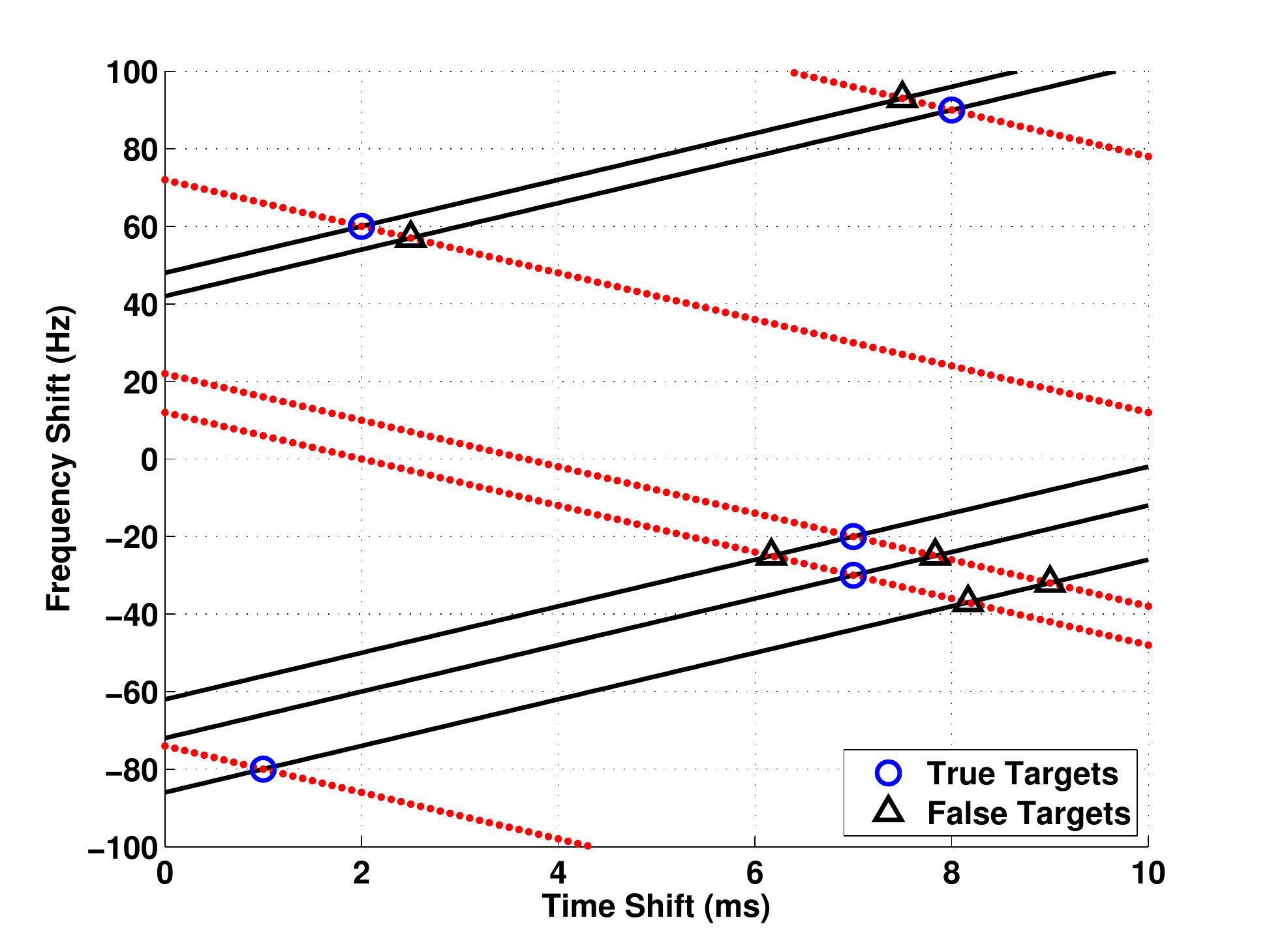}
    \caption{$f_c = \pm 3000$ Hz/s.}
    \label{fg:LFM-2k-BW}
  \end{subfigure}
  \begin{subfigure}{0.49\columnwidth}
    \centering
    \includegraphics[width=0.99\columnwidth]{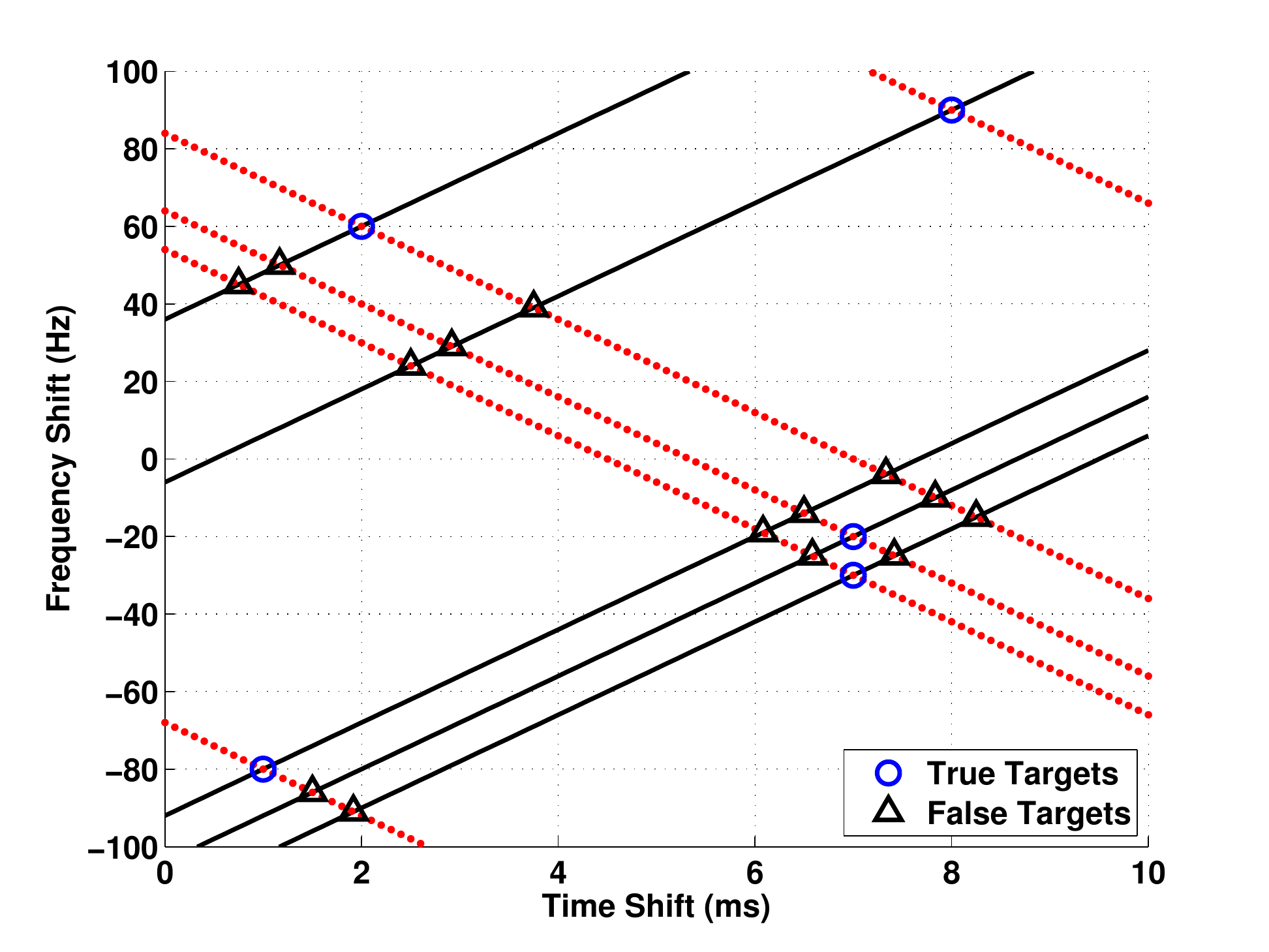}
    \caption{$f_c = \pm 6000$ Hz/s.}
    \label{fg:LFM-4k-BW}
  \end{subfigure}
  \caption{The constraints imposed by four LFM pulses for an example scene with five targets, $\tau_{\max} = 0.01$s and $f_{\max} = 100$Hz.  The lines with positive slope enforce constraints from the positive chirp, while the lines with negative slope enforce constraints from the negative chirp.  The line intersections show all possible time--frequency shift pairs that can explain the recovered frequencies.  Both sets of constraints are needed to disambiguate the true targets.}
  \label{fg:lfm-constraints}
\end{figure}

The recovery is more complicated for multiple targets.  For example, Fig. \ref{fg:LFM-2k-BW} shows the set of constraints provided by four LFM pulses ($M=4$) when there are five targets in the received signal.  Each point of intersection between any two lines satisfies the constraints imposed by the pair of LFM pulses.  
Recovery requires the use of the phase information and more LFM pulses with different chirp rates.  The procedure to recover multiple targets, which is one of the main novelties of this paper, is presented in the next subsection.

\subsection{Resolving Ambiguity Between LFM Pulses}\label{sec:matching-pairs}
The recovered sinusoid parameters from each pulse are
$\psi^{m}_k = \phi_k + f^m_c\tau_k^2$ and
$\nu^{m}_k = f_k - 2f^m_c\tau_k$
for targets $k = 1,\ldots, K$ and pulses $m = 1,\ldots,M$.  This is a system of equations in which each target contributes three unknowns ($\phi_k$, $\tau_k$, and $f_k$) and each pulse contributes one linear constraint and one quadratic constraint.  We will first concentrate on the frequency $\nu_k^m$ because it contains both the time shift and frequency shift and is linear.  For this constraint, each target contributes two unknowns ($f_k$ and $\tau_k$) while each pulse contributes a single constraint.

Collect the unknown parameters into a column vector
$$\mathbf{\beta} = [f_1,\ldots,f_K, \tau_1,\ldots, \tau_K]^T$$
and the recovered frequencies for the $m^{th}$ pulse into a column vector
$$\mathbf{\nu}^m = [\nu_1^m,\ldots,\nu_K^m]^T.$$
The constraints from the $m^{th}$ pulse can be written via the matrix equation
\begin{equation}\label{eq:single-pulse-frequency-constraint-matrix}
  \mathbf{A}^m\mathbf{\beta} = \mathbf{\nu}^m
\end{equation}
with the matrix
$$\mathbf{A}^m = \begin{bmatrix} \mathbf{I}_K & -2f_c^m \mathbf{I}_K\end{bmatrix} = \begin{bmatrix} 1 & -2f_c^m\end{bmatrix} \otimes \mathbf{I}_k$$
representing the constraints and where $\otimes$ denotes the Kronecker product.  The $K\times K$ identity matrix is denoted $\mathbf{I}_K$.  The matrix equation \eqref{eq:single-pulse-frequency-constraint-matrix} is underdetermined as there are $2K$ unknowns and $K$ equations.  We can add more constraints by sending more LFM pulses.  Each LFM pulse contributes a set of $K$ constraints of the form \eqref{eq:single-pulse-frequency-constraint-matrix}, while $\mathbf{\beta}$ must satisfy all sets of constraints simultaneously.  However, the frequencies $\nu_k^m$ are unordered and the relationship between the pulses is not immediate.  In other words, if the recovered frequencies are sorted in ascending order so that $\nu_1^m \leq \nu_2^m \leq \ldots \leq \nu_K^m$, then $\nu_k^m$ may not have been produced by the same target as $\nu_k^p$ for $m\neq p$.  This is most easily understood by examining the plot in Fig. \ref{fg:LFM-2k-BW}.  The solid black lines show the constraints imposed by an LFM pulse with $f_c = 3000$ Hz/s, and the dashed red lines show the constraints imposed by an LFM pulse with $f_c = -3000$ Hz/s.  We know that each target produced a single frequency, so any set of 5 intersecting points satisfying this constraint could explain the received signal.

We formalize this process by putting the constraints \eqref{eq:single-pulse-frequency-constraint-matrix} into a single equation.  The $MK \times 2K$ matrix $\mathbf{A}$ contains the constraints from each pulse and the $MK \times 1$ vector $\mathbf{\nu}$ contains the recovered frequencies from each pulse:
$$\mathbf{A} = \begin{bmatrix} \mathbf{A}^1 \\ \vdots \\ \mathbf{A}^M \end{bmatrix}
  = \begin{bmatrix} 1 & -2f_c^1 \\ \vdots & \vdots \\ 1 & -2f_c^M \end{bmatrix} \otimes \mathbf{I}_K = \mathbf{B} \otimes \mathbf{I}_K \ \ \ \text{,} \ \ \ \mathbf{\nu} = \begin{bmatrix} \mathbf{\nu}^1 \\ \vdots \\ \mathbf{\nu}^M\end{bmatrix}.$$
The parameter vector satisfies 
\begin{equation}\label{eq:frequency-linear-system}
  \mathbf{A}\beta = \mathbf{P}\mathbf{\nu}
\end{equation}
where the matrix $\mathbf{P} = \text{diag}(\mathbf{I}_K, \mathbf{P}_2, \ldots, \mathbf{P}_M)$ is a block diagonal matrix where $\mathbf{P}_m$ are permutation matrices that account for the unordered $\nu^m$s and match each frequency with the target that produced it.

As an illustrative example, consider Fig.~\ref{fg:LFM-2k-BW}.  The figure shows the linear constraints imposed by the frequency relations in \eqref{eq:pos-neg-freq-phase} for five targets (shown as blue circles).  Without knowledge of the true target locations (in time and frequency shift space), any set of five intersections between the red and black lines that also satisfies $\tau \leq \tau_{\max}$ could explain the recovered frequencies. 
In this example scene, there are a total of four valid explanations of the recovered frequencies.
The pairs of lines with no ambiguities can be matched with no further information.
For example, the true target at a time shift of $1$ms and frequency shift of $-80$Hz can be matched.
The ambiguous matchings are resolved through use of the recovered phase terms and a second LFM pulse.
We first explain how the phase information can be used.

Consider the case of two LFM pulses ($M=2$) and let $f_c^2 = -f_c^1$.  Index the frequencies for the first pulse with $k$ and for the second pulse with $\ell$ and note two relations:
\begin{align}
  \frac{\nu^{1}_k - \nu^{2}_{\ell}}{-4 f_c} &= \frac{f_k - f_{\ell}}{-4f_c^1} + \frac{\tau_k + \tau_{\ell}}{2}\label{eq:nu-expanded} \\
  \frac{\psi^{1}_k - \psi^{2}_{\ell}}{2f_c^1} &= \frac{\phi_k - \phi_{\ell}}{2f_c^1} + \frac{\tau_k^2 + \tau_{\ell}^2}{2} \label{eq:psi-expanded}
\end{align}

If $k$ and $\ell$ correspond to the same target, then $f_k = f_{\ell}$, $\tau_k = \tau_{\ell}$, $\phi_k = \phi_{\ell}$, and the frequencies and the phases satisfy
\begin{equation}\label{eq:omega-difference}
  \frac{\nu^{1}_k - \nu^{2}_k}{-4 f_c} = \tau_k \ \ \text{and} \ \ \frac{\psi^{1}_k - \psi^{2}_k}{2f_c^1} = \tau_k^2
\end{equation}
%
where both quantities only depend on $\tau_k$.  We check if the following relationship is satisfied
\begin{equation}\label{eq:phase-matching-condition}
  \left(\frac{\nu_k^1 - \nu_{\ell}^2}{-4 f_c^1}\right)^2 = \frac{\psi_k^1 - \psi_{\ell}^2}{2f_c^1}.
\end{equation}
If this condition holds true, then we declare that $k$ and $\ell$ describe the same target.  Otherwise, we continue checking.  This procedure is summarized in lines 1-10 of Algorithm \ref{alg:delay-Doppler-recovery}.

There is one more source of ambiguity caused by the unknown target phases $\phi_k$.  The \emph{hypothesized} time shift $\tau_h(k,\ell)$ is the time shift that results from \eqref{eq:omega-difference}
\begin{align}
  \tau_h(k,\ell) &= \frac{\nu^{1}_k - \nu^{2}_{\ell}}{-4 f_c^1} = \frac{f_{k}-f_{\ell}}{-4f_c^1} + \frac{\tau_k + \tau_{\ell}}{2}. \label{eq:candidate-tau}
\end{align}
The relationship \eqref{eq:phase-matching-condition} can be satisfied even if $k$ and $\ell$ do not correspond to the same target.  Substitute \eqref{eq:psi-expanded} into \eqref{eq:phase-matching-condition}
\begin{equation*}
  \tau_h(k,\ell)^2 = \frac{\psi^{1}_k - \psi^{2}_{\ell}}{2f_c^1} = \frac{\phi_k - \phi_{\ell}}{2f_c^1} + \frac{\tau_k^2 + \tau_{\ell}^2}{2},
\end{equation*}
and rearranging gives the condition on $\phi_k-\phi_{\ell}$ for ambiguity
\begin{equation*}
  \phi_k - \phi_{\ell} = 2f_c^1\tau_h(k,\ell)^2 - f_c^1\left(\tau_k^2 + \tau_{\ell}^2\right).
\end{equation*}

However, notice that if there is an ambiguity from the solution of \eqref{eq:frequency-linear-system} then there will be at least one other \emph{false target} in addition to the one just described.  In this case, target $k$ will contribute to the frequency and phase of the second LFM pulse ($m=2$) and target $\ell$ will contribute to the first LFM pulse ($m=1$).  The condition on the target phases in this case is
\begin{equation*}
  \phi_k - \phi_{\ell} = -2f_c^1\tau_h(\ell,k)^2 + f_c^1\left(\tau_k^2 + \tau_{\ell}^2\right).
\end{equation*}
We can use \eqref{eq:candidate-tau} to find the condition under which both of these conditions are ambiguous:
$|f_k-f_{\ell}| = 2f_c^1|\tau_k-\tau_{\ell}|,$
which is resolved if another pair of LFM pulses is sent with a different chirp rate $f_c^3 \neq f_c^1$.

\subsection{Recovery of Target Amplitude and Phase}\label{sec:dopp-delay-est}
With the time and frequency shifts recovered and ambiguities resolved,
the target phase $\phi_k$ is recovered via
$\hat{\phi}_k = \frac{1}{2}(\hat{\psi}^{1}_k + \hat{\psi}^{2}_k)$
and the amplitude is recovered via \eqref{eq:least-squares-amplitude}: $|\hat{c}_k| = |\hat{\zeta}_k|$.

\begin{algorithm}
  \caption{Algorithm for recovering the time and frequency shifts from the measurements}
  \label{alg:delay-Doppler-recovery}
  \begin{algorithmic}[1]
    \STATE Data: $\tilde{y}_m[n]$
    \STATE Use Algorithm \ref{alg:frequency-phase-recovery} to find the frequencies $\hat{\nu}_k^m$ and phases $\hat{\psi}_k^m$ for $k = 1,\ldots,K$ and $m=1,\ldots,M$
    \STATE Find possible target time and frequency shifts by solving \eqref{eq:single-pulse-frequency-constraint-matrix} for $m=1,2$
    \FOR{$k=1,...,K$}
      \FOR{$\ell=1,...,K$}
        \STATE Calculate delay hypothesis $\tau_h(k,\ell)$.
        \IF {$0 \leq \tau_h(k,\ell) \leq \tau_{\max}$ and $\tau_h(k,\ell)^2 = \frac{\psi_k^1 - \psi_{\ell}^2}{2f_c^1}$}
        \STATE Declare $k$ possibly matched with $\ell$
        \ENDIF
      \ENDFOR
    \ENDFOR
    \STATE Detect ambiguous matchings where a sinusoid $k$ has multiple matchings $\ell$, and vice-versa. Repeat 1-10 for $m=3,4$ to resolve the ambiguities.
    \STATE Calculate the recovered frequency shift $\hat{f}_k = \frac{1}{2}\hat{\nu}^{1}_k + \hat{\nu}^{2}_k.$
    \STATE Calculate the recovered time shift $\hat{\tau}_k = \frac{1}{4f_c^1}\hat{\nu}^{1}_k - \hat{\nu}^{2}_k.$
    \STATE Calculate the recovered target phase $\hat{\phi}_k = \frac{1}{2}\hat{\psi}^{1}_k + \hat{\psi}^{2}_k.$
    \STATE Calculate the recovered target amplitude $|\hat{c}_k| = |\hat{\zeta}_k|.$
  \end{algorithmic}
\end{algorithm}

\subsection{Sufficient Conditions for Perfect Noiseless Recovery}
The following theorem, first presented in \cite{harms-camsap-2013}, establishes that we can perfectly recover the time and frequency shifts if we transmit four LFM pulses, i.e., $M=4$.

\begin{theorem}[Perfect Recovery of Time Shifts and Frequency Shifts]\label{thm:perfect-recovery-noiseless}
  For a given $\tau_{\max}$ and $f_{\max}$ and $M=4$, choose $f_c^2 = -f_c^1$ and $f_c^4 = -f_c^3$ such that $f_c^1 \neq f_c^3$ and satisfying Lemma~\ref{lemma:ambiguous-phase} in Appendix \ref{sec:ambiguous-phase-terms}.  Take the samples $\tilde{y}_m[n]$ for $m=1,2,3,4$ with $f_s$ satisfying \eqref{eq:sampling-rate} and the number of samples per pulse $N$ satisfying
$N \geq K + 1.$
Then Algorithms \ref{alg:frequency-phase-recovery} and \ref{alg:delay-Doppler-recovery} perfectly recover the time shifts $\tau_k$, frequency shifts $f_k$, and (complex) scalings $c_k$.

\end{theorem}

\begin{proof}
Algorithm~\ref{alg:frequency-phase-recovery} offers perfect reconstruction of the frequencies and phases in the absence of any noise in the measurements \cite{tufts-kumaresan-1982}.  Note that the matrix $\mathbf{Y}$, and subsequently $\mathbf{R}$, have rank $K$ in this case.  

The matching procedure in Algorithm \ref{alg:delay-Doppler-recovery} relies on $\nu_k^m$ not being aliased and $\theta^m(\tau)$ being bijective.  The restriction on $f_s$ ensures the former and Lemma~\ref{lemma:ambiguous-phase} ensures the latter.

%

\end{proof}

\subsection{Resource Usage}
We now provide a brief discussion of resource usage to highlight the advantages of our approach.  We concentrate on two main points of resource usage: 1) the rate of samples and number of samples required for recovery and 2) the time-bandwidth product of the waveform used in recovery of the delay and Doppler.  

\subsubsection{Sampling rate and number of samples}
The sampling rate must satisfy the Nyquist condition \eqref{eq:sampling-rate} to avoid anti-aliasing.  The lower bound on $f_s$ depends on $f_{\max}$ and $\tau_{\max}$, and an increase in either requires a corresponding increase in $f_s$.  We first compare this sampling rate to the rate required to sample the unprocessed LFM pulse (i.e., without dechirping).  The bandwidth of an LFM pulse is approximately
$W \approx 1/T_p + f_cT_p$
meaning that processing of this signal directly (e.g., with cross-correlation processing) would require a sampling rate proportional to $W$.  In contrast, our approach requires a sampling rate \eqref{eq:sampling-rate} of
$f_s \geq 2(f_{\max} + 2f_c\tau_{\max}).$
Recall that $T_o = T_p - \tau_{\max}$ meaning $T_p > \tau_{\max}$.  To satisfy the assumption of fixed parameters over the measurement interval, we require that $f_{\max} < \frac{\eta}{T_p}$ as well where $\eta$ is a fixed constant set by the physics of the scenario and generally $\eta<1$.  Our approach therefore requires a lower sampling rate than a direct sampling of the LFM pulse.

In addition to the lower bound on the sampling rate, the frequency estimation step requires a minimum number of samples, captured in Theorem~\ref{thm:perfect-recovery-noiseless}, to recover the $K$ frequencies associated with each target.  In the case of uncorrupted measurements, the KT algorithm requires $N \geq K + 1$ measurements to recover $K$ distinct frequencies.  Recall that the measurements used in the recovery algorithm are taken over the measurement interval $T_o$, so we have a lower bound on the measurement interval
\begin{equation}\label{eq:measurement-time}
  T_o \geq \frac{N}{f_s}.
\end{equation}

Further, the transmitted pulse duration $T_p$ and pulse repetition rate $T$ (of each pulse) are set based on $T_o$.  First, choose a $T_o$ satisfying \eqref{eq:measurement-time}, and then choose
$$T_p = T_o + \tau_{\max} \geq \frac{N}{f_s} + \tau_{\max}$$
and
$$T = T_p + T_g \geq \frac{N}{f_s} + 2\tau_{\max}$$
where $T_g \geq \tau_{\max}$.  The total time processing time, denoted $\mathcal{T}$, needed for $M$ pulses is then $\mathcal{T} = M\cdotp T$.

\subsubsection{Time-bandwidth Product}
Given the previous discussion of sampling rate and number of samples, the total time needed to recover the LTV characterization is $\mathcal{T} = M\cdotp T$.
The bandwidth of an LFM pulse is approximately
$W \approx \frac{1}{T_p} + f_cT_p$s,
so that the time-bandwidth product, $\mathcal{T}\cdotp W$, of a single LFM pulse is approximately
$\mathcal{T}\cdotp W \approx M\cdotp T\cdotp \left(1/T_p + f_cT_p\right),$
which means that $\mathcal{T}\cdotp W$ must scale with $1+K^2$.

\subsubsection{Computational Complexity}
On first inspection, the appearance of the permutation matrix in \eqref{eq:frequency-linear-system} makes it look like the complexity of finding a solution to the matching problem is factorial in $K$.  However, note that many of the \emph{possible solutions} will actually not be feasible given the $\tau_{\max}$ and $f_{\max}$ constraints.  Also, note that once a \emph{true target} is found, then the frequency from each pulse corresponding to that target can be eliminated from the search, effectively fixing that portion of the permutation matrix and shrinking the size of the problem.  The problem solution can then be found by performing comparisons of parameters from pulse to pulse that is polynomial in $K$.

\subsection{Discussion and Comparison}\label{sec:noiseless-discussion}
We now place this work in context of other recent advances in LTV system characterization.  In \cite{bajwa11-radar} and \cite{friedlander-parametric2012}, sequential processing is employed.  In the case of \cite{bajwa11-radar}, the delays are first recovered followed by Dopplers at each recovered delay; in the case of \cite{friedlander-parametric2012}, the Dopplers are first recovered followed by the delays at each recovered Doppler.  One disadvantage to such a sequential approach is that errors in the first stage propagate through to the second stage.  In contrast, our approach in this paper requires only one recovery stage so errors do not propagate.  Both of these techniques also require the transmission of a series of pulses.  In the case of \cite{friedlander-parametric2012}, these pulses are stepped-frequency pulses, quite similar in spirit to an LFM pulse.  Because a sequential technique is employed, an assumption must be made on the maximum number of delays associated with any single Doppler.  The required number of samples for recovery depends on the number of distinct delays associated with each Doppler.  As a worst-case scenario, if there are $K$ targets parameterizing an LTV system, then the minimum number of samples needed to describe the system is on the order of $K^2$.  In contrast, the approach described in this paper requires the minimum number of samples to be on the order of $K$.

The work \cite{eldar-doppler-focusing} does not use a sequential technique, but instead uses a technique termed \emph{Doppler focusing} where multiple pulses are coherently processed to find delays at the Dopplers which are \emph{in focus}, meaning Dopplers in a small interval around a nominal center.  The approach uses low-rate samples, i.e., below the bandwidth of the pulses, where the number of samples depends linearly on the number of targets $K$.  These low-rate samples are used to find the low-pass Fourier series coefficients of the received pulses.  These coefficients are shown to be samples of a sum of sinusoids determined by the delays, Dopplers, and amplitudes of the targets.  Parametric recovery techniques can then be used to recover the delays present in the \emph{focus zone}.  Alternatively, compressed sensing-based techniques can be used if the delays are assumed to lie on a grid.  One disadvantage of this approach is the basis mismatch that occurs if the actual delays do not lie on the assumed grid \cite{chi11}.  Another disadvantage of this approach results from the Doppler focusing.  The width of the \emph{focus zone} is proportional to $\frac{1}{MT}$ where $M$ is the number of pulses and $T$ is the pulse repetition interval.  In the absence of any prior knowledge about Doppler locations, a minimum number of focus zones must be used, and delays calculated from each, to ensure that all Dopplers are covered.  Each focus zone also has a resolution proportional to $\frac{1}{MT}$ meaning that all Dopplers in this interval will be associated with the center Doppler.  Our approach described in this paper avoids this problem by recovering the delays and Dopplers directly.

\section{Identifying LTV Operators from Noise-corrupted Samples}\label{sec:noisy-estimation}
The likely more interesting situation in most applications is the case of noisy samples, i.e., $\tilde{\varepsilon}_m[n] \neq 0$.  In this case, we can use the procedure described above with some modifications to make it more robust to noisy measurements.  To start, we assume the noise samples $\varepsilon_m[n]$ are i.i.d. Gaussian, so the dechirped noise samples $\tilde{\varepsilon}_m[n]$ are as well.  We use Algorithms \ref{alg:frequency-phase-recovery} and \ref{alg:delay-Doppler-recovery} with small modifications described below that are more robust to noisy measurements.  Due to the presence of noise, the estimate of the frequency and phase will contain some error.  However, the frequency estimator produces consistent estimates with variance that approaches the Cramer-Rao lower bound (CRLB).  The variance of the estimate tells us about our ability to resolve closely spaced frequencies, and by extension targets.  If two (or more) estimates are too close to resolve, then the matching procedure cannot reliably distinguish between them.  Where, in the noiseless case, we had an ambiguity equality (for a single pulse), in the noisy case we have an ambiguity interval.  The processing of further LFM pulses shrinks the ambiguity interval to refine the estimates and decrease the estimator variance (increase the estimator resolution), asymptotically achieving zero variance.  Finally, we show that the estimator of the time shifts, frequency shifts, and amplitudes converges asymptotically to the true parameters, regardless of the parameter values.  We further argue that for most scenes of interest (i.e., locations of targets), we can achieve resolution proportional to the noise power from a finite number of samples and pulses.

\subsection{Denoising via Atomic Norm Minimization}\label{sec:denoising}
The KT algorithm described in the previous section, as well as many other spectral estimation algorithms, is sensitive to SNR in the sense that there is typically a threshold SNR, below which, the recovery breaks down.  To improve the performance of the algorithm at lower input SNR, we propose using a denoising step in the recovery.  The procedure we propose uses atomic norm minimization to pre-process the noisy samples to increase the SNR of the samples fed into the spectral estimation stage.  The errors that result from the denoising procedure are not necessarily independent or Gaussian, so we provide justification by examining the empirical first and second order statistics of the denoised samples through Monte Carlo simulations.
We begin with a brief description of the atomic norm and its semidefinite characterization before describing the denoising procedure.

We start by defining the set of atoms
\begin{align*}
  \mathcal{A} = &\{a(t; \phi,\tau,f) : \tau \in [0,\tau_{\max}), f\in (-f_{\max},f_{\max}), \phi \}
\end{align*}
where
$a(t; \phi,\tau,f) = e^{j\phi + j2\pi f t}p(t-\tau)$
with an (almost) arbitrary pulse $p(t)$.  We assume the atoms are contained in $L^2$, and the notation emphasizes that they are functions of $t$ and parameterized by a phase $\phi$, time shift $\tau$, and frequency shift $f$.  These atoms are, not coincidentally, also the building blocks of the LTV model \eqref{eq:LTV-system}.  Note that $c_k = |c_k|e^{j\phi_k}$ so that \eqref{eq:LTV-system} can be written as a superposition of elements from $\mathcal{A}$ scaled by real, positive coefficients $|c_k|$, or we can subsume the phase term into the complex coefficients and set $\phi = 0$
\begin{align*}
  y(t) &= \sum_{k=1}^K |c_k| a(t; \phi_k, \tau_k, f_k) = \sum_{k=1}^K c_k a(t; 0,\tau_k,f_k).
\end{align*}

The atomic norm \cite{recht12-atomic-norm,bhaskar13-atomic-norm,tang13-offgrid} of a signal $y(t)$ (relative to a set of atoms $\mathcal{A}$) is
\begin{align*}
  ||y(t)||_{\mathcal{A}} &= \inf\{\gamma > 0 : y \in \gamma\cdotp\text{conv}(\mathcal{A})\} \\
  &= \inf_{\mathbf{c}_k, \mathbf{\tau}_k, \mathbf{f}_k} \left\{\sum_{k=1}^K |c_k| : y(t) = \sum_{k=1}^K c_k a(t; 0,\tau_k,f_k)  \right\}
\end{align*}
where the variables in the infimum satisfy $|c_k| \geq 0, \tau_k \in [0,\tau_{\max}), f_k\in [-f_{\max},f_{\max}]$ and $\text{conv}(\cdot)$ is the convex hull.  The corresponding dual norm is
$$||z(t)||^*_{\mathcal{A}} = \sup_{a\in\mathcal{A}} \langle z(t),a \rangle = \sup_{\phi, \tau, f} \langle z(t),a(t; \phi,\tau,f)\rangle.$$
Additionally, we define a set of sampled atoms as
$\mathcal{A}_s = \{a_s(n; \phi,\tau,f) : \tau \in [0,\tau_{\max}), f\in (-f_{\max},f_{\max}), \phi \in [0,2\pi) \}$
where
$a_s(n; \phi,\tau,f) = e^{j\phi + j2\pi f nT_s}p(nT_s-\tau)$,
with $n = 0,...,N-1$, are vectors in $\mathbb{C}^N$, and $T_s$ is the sampling interval.  The sampled atomic set is useful because we want to denoise sampled measurements.

\subsubsection{Dechirped and Sampled LFM Pulses}
We are interested in the dechirped and sampled LFM pulses $\tilde{y}_m[n]$.  These signals can be built using the sampled atomic set with
$p(nT_s-\tau) = e^{j2\pi(f_c^m\tau^2 - f_0^m \tau)}e^{-j2\pi (2f_c^m\tau) nT_s}$
so that
$a_s(n; \phi,\tau,f) = e^{j\phi}e^{j2\pi(f_c^m\tau^2 - f_0^m \tau)}e^{j2\pi (f - 2f_c^m\tau) nT_s}$
are sampled sinusoids with frequency dependent on $\tau$ and $f$ and phase dependent on $\phi$ and $\tau$.  As before, we can change the parameters to $\nu$ and $\psi$ so the atoms are
$a_s(n; \psi, \nu) = e^{j\psi}e^{j2\pi \nu nT_s}.$
The measurements are
\begin{equation}\label{eq:corrupted-measurements}
  \tilde{y}_m[n] = \sum_{k=1}^K |c_k|a_s(n; \psi_k, \nu_k) + \tilde{\varepsilon}_m[n]
\end{equation}
where $\tilde{\varepsilon}_m[n]$ is independent, white Gaussian noise.

\subsubsection{Denoising}
We are now in a position to leverage a result from Bhaskar et al. \cite{bhaskar13-atomic-norm} to denoise measurements of the form \eqref{eq:corrupted-measurements}.  Consider $\tilde{\varepsilon}_m[n]$ independent Gaussian with variance $\sigma^2$ and samples $\tilde{y}_m[n]$, $n=0,...,N-1$.  We solve the following \emph{atomic soft thresholding} (AST) problem.

\begin{equation}\label{eq:optimization-atomic-norm}
  \hat{y} = \arg\min_{z} \frac{1}{2}||z[n] - \tilde{y}_m[n]||_2^2 + \eta||z[n]||_{\mathcal{A}_s}
\end{equation}
where $\eta$ is a regularization parameter controlling the relative impact of the mean-squared error and the atomic norm.

The following theorem captures the denoising performance for the estimate $\hat{y}$ of the AST problem for the case of Gaussian noise \cite[Theorem 2]{bhaskar13-atomic-norm}.

\begin{theorem}[AST with Gaussian noise]
\label{thm:AST-asymptotic-error}
Let $y^{\star}[n] = \sum_{k=1}^K g_k e^{j2\pi n \omega_k}$ for complex numbers $g_1,..,g_K$ and unknown normalized frequencies $\omega_1,...,\omega_K \in [0,1]$.  Consider measurements given by $y[n] = y^{\star}[n] + w[n]$ where $w[n] \sim \mathcal{N}(0,\sigma^2 I_N)$ i.i.d. The estimate $\hat{y}[n]$ obtained by solving \eqref{eq:optimization-atomic-norm} with $\eta = \sigma\sqrt{N\log N}$ has asymptotic mean-squared-error
\begin{equation}\label{eq:asymptotic-error}
  \frac{1}{N}\mathbb{E}||\hat{y} - y^{\star}||_2^2 < \sigma \sqrt{\frac{\log N}{N}}\sum_{k = 1}^K|g_k|
\end{equation}
for $N$ sufficiently large (see \cite{bhaskar13-atomic-norm} for details).
\end{theorem}

We apply Theorem \ref{thm:AST-asymptotic-error} to our corrupted samples $\tilde{y}_m[n]$ by equating the frequencies
$\omega_k = \nu_kT_s$
and the complex coefficients with the amplitude and phase
$g_k = |c_k|e^{j\psi_k}$
and by letting $\tilde{\varepsilon}_m[n] \sim \mathcal{N}(0,\sigma^2 I_N)$ i.i.d.  If we write the uncorrupted measurements as
$$y_m^{\star}[n] = \sum_{k=1}^K |c_k|a_s(n; \psi_k, \nu_k)$$
and the denoised measurements as $\hat{y}_m[n]$, then the error $e_m[n] = \hat{y}_m[n] - y_m^{\star}[n]$ satisfies
\begin{equation}\label{eq:error-variance}
  \frac{1}{N}\mathbb{E}||e_m[n]||_2^2 < \sigma \sqrt{\frac{\log N}{N}}\sum_{k = 1}^K|c_k|,
\end{equation}
for sufficiently large $N$, and we write the denoised measurements as
$\hat{y}_m[n] = y_m^{\star}[n] + e_m[n].$
Note that $e_m[n]$ is not guaranteed to be independent Gaussian, but empirically is zero mean with variance \eqref{eq:error-variance} and tends to follow Gaussian statistics.

The optimization problem \eqref{eq:optimization-atomic-norm} relies on calculating the atomic norm of $z[n]$.  This can be accomplished via a semidefinite programming (SDP) problem \cite{bhaskar13-atomic-norm} so that the solution of \eqref{eq:optimization-atomic-norm} is equivalent to the solution of
\begin{align*}
  \min_{t,u,x} &\frac{1}{2}||x-\tilde{y}[n]||_2^2 + \frac{\gamma}{2}(t+u_1) \ \
  \text{subject to} \ \ \begin{bmatrix} T(u) & x \\ x^* & t \end{bmatrix} \succeq 0
\end{align*}
where the function $T(u)$ forms a Toeplitz Hermitian matrix from the entries in $u$ (i.e., $u_1$ on the diagonal, $u_N$ in the upper right corner, etc.).  The solution $\hat{y}$, which satisfies \eqref{eq:asymptotic-error}, is the argument of the vector $x$ from the optimization.
The denoised measurements are then provided to modified versions of Algorithms 1 and 2.

\subsection{Frequency and Phase Recovery}
Taking the denoised measurements as input, we make two standard modifications to Algorithm \ref{alg:frequency-phase-recovery} that better handle noise in the measurements.  The first change is replacement of the correlation matrix $\mathbf{R}$ with the \emph{proxy} correlation matrix $\mathbf{\tilde{R}}$.  The proxy correlation matrix is computed by first finding the singular value decomposition (SVD) of the correlation matrix $\mathbf{R} = \mathbf{USV}^H$.  The proxy is then formed from
$\mathbf{\tilde{R}} = \mathbf{\tilde{U}\tilde{S}\tilde{V}}^H$
where $\mathbf{\tilde{S}}$ contains only the $K$ largest singular values and $\mathbf{\tilde{U}}$ and $\mathbf{\tilde{V}}$ contain the corresponding $K$ singular vectors.  In the absence of noise, $\mathbf{R}$ is of rank $K$ and therefore only contains $K$ non-zero singular values.  In the presence of noisy measurements, $\mathbf{R}$ is full rank but the singular values cluster into two groups: the largest $K$ contain most of the information about the constituent frequencies while the remaining are close to zero and contain information about the noise.  Using only the $K$ largest reduces the effect of the noise \cite{tufts-kumaresan-1982}.

The second change occurs when finding the $K$ roots of the prediction filter.  In the absence of noise, $K$ of the roots will be on the unit circle while the rest reside inside the unit circle.  When noise is present, the $K$ roots we are after will likely not lie exactly on the unit circle.  We therefore search for the $K$ roots that are closest to the unit circle \cite{tufts-kumaresan-1982}.  Note that if $K$ is unknown \emph{a priori}, which is often the case in practice, $K$ can be estimated from the clustering of the larger singular values of $\mathbf{R}$ discussed above.  

\begin{algorithm}
  \caption{Algorithm for recovering the frequency and phase of sinusoids from corrupted measurements}
  \label{alg:noisy-frequency-phase-recovery}
  \begin{algorithmic}[1]
    \STATE Data: $\tilde{y}_m[n]$
    \STATE Calculate proxy correlation matrix $\mathbf{\tilde{R}}$ from the $K$ largest singular values of $\mathbf{Y}^H\mathbf{Y}$.
    \STATE Calculate the coefficients $\mathbf{h} = -\mathbf{\tilde{R}}^{-1}\mathbf{Y}^H\mathbf{y}$.
    \STATE Find the $K$ roots $\hat{z}_k$ of $H(z)$ closest to the unit circle.
    \STATE Calculate $\hat{\nu}_k = \frac{T_s}{2\pi}\text{phase}(\hat{z}_k)$.
    \STATE Calculate $\hat{\psi}_k$ and $|\hat{c}_k|$ from the least-squares solution $\hat{\beta}_k$.
  
  \end{algorithmic}
\end{algorithm}

The variance of the frequency estimate approaches the asymptotic Cramer-Rao bound, summarized in the following lemma (Theorem 4.1 of \cite{stoica-89}).

\begin{lemma}\label{lemma:KT-CRB}
Given $N$ uniform samples of a signal consisting of a superposition of sinusoids, e.g., given by the model \eqref{eq:measurement-samples}, corrupted by independent circularly symmetric complex Gaussian noise with variance $\sigma^2$, the KT algorithm (Algorithm \ref{alg:noisy-frequency-phase-recovery}) produces a consistent estimate of the frequencies $\hat{\nu}_k$, $k=1,\ldots,K$, that is asymptotically efficient, that is, the variance of the estimate approaches, as $N\to\infty$, the asymptotic Cramer-Rao bound
\begin{equation}\label{eq:CRD-freq-asympotic}
  \lim_{N\to\infty}\text{var}(\hat{\nu}_k)\cdot N^3 = \frac{6}{\text{SNR}_k T_s^2}
\end{equation}
where $\text{SNR}_k = \frac{|c_k|^2}{\sigma^2}$ is the signal to noise ratio for the $k^{th}$ target.
Additionally, the estimate of the amplitude $\hat{c}_k$ is consistent and asymptotically efficient in the sense that it also approaches the asymptotic Cramer-Rao bound
$$\lim_{N\to\infty}\text{var}(\hat{c}_k)\cdot N = \sigma^2.$$

\end{lemma}

Lemma \ref{lemma:KT-CRB} tells us that the estimates of the frequencies and amplitudes converge to the true values ($\hat{\nu}_k \to \nu_k$ and $\hat{c}_k \to c_k$ as $N\to\infty$) and that the rate of convergence is close to the Cramer-Rao bound, at least to first order.  The variance of the frequency estimate decays, to first order, as
\begin{equation}\label{eq:CRB-freq}
  \text{var}(\hat{\nu}_k) \approx \frac{6}{N^3 \text{SNR}_k T_s^2},
\end{equation}
and the variance of the amplitude decays as
$$\text{var}(\hat{c}_k) \approx \frac{\sigma^2}{N}.$$
We use these convergence rates as a rough measure of the resolution offered by the KT algorithm.


\subsection{Recovery of Time Shifts and Frequency Shifts}
Note that the uncertainty introduced into the estimate of the phases $\psi_k^m = \phi_k + f_c^m\tau_k^2$ causes trouble when trying to recover $\tau_k$.  Because the relationship is quadratic, error introduced into $\hat{\psi}_k^m$ affects recovery of smaller values of $\tau_k$ more than larger ones.  In fact, as $\tau_k \to 0$, the SNR $\to 0$ as well.  We therefore do not use the phase information for recovery from noisy measurements.

Uncertainty in the estimated frequencies affects the relation \eqref{eq:frequency-linear-system} by introducing an \emph{error} term $\gamma \in \mathbb{C}^{M\cdotp K}$ into the recovered frequencies $\mathbf{\hat{\nu}} = \mathbf{\nu} + \mathbf{\gamma}$.  The vector $\mathbf{\nu}$ contains the true frequencies, and because the estimator in Algorithm~\ref{alg:noisy-frequency-phase-recovery} is consistent (see Lemma~\ref{lemma:KT-CRB}), we consider $\mathbf{\gamma}$ to be a zero mean vector with independent entries and variances \eqref{eq:CRB-freq}.
The following theorem establishes that if $\mathbf{\gamma}$ contains errors of finite variance and the chirp rates of the transmitted LFM pulses are different, then solving a least-squares problem produces an exact solution $\mathbf{\hat{\beta}}$ asymptotically.

We first remind the reader of the problem formulation.  With noisy measurements, the accuracy of the estimate depends on the number of LFM pulses used, $M$, because the number of measurements increases with $M$.  We make this explicit by writing all vectors and matrices that have size dependent on $M$ as a function of $M$, e.g., $\mathbf{\hat{\nu}}(M)$ and $\mathbf{A}(M)$.  The vector of estimated frequencies is $\mathbf{\hat{\nu}}(M) = [\mathbf{\hat{\nu}}^1 \cdots \mathbf{\hat{\nu}}^M]^T$ and $\mathbf{\hat{\nu}}(M) = \mathbf{\nu}(M) + \mathbf{\gamma}(M)$ where $\mathbf{\nu}(M)$ contains the true frequencies and $\mathbf{\gamma}(M)$ contains the errors.  The matrix
$$\mathbf{A}(M) = \begin{bmatrix} 1 & -2f_c^1 \\ \vdots & \vdots \\ 1 & -2f_c^M \end{bmatrix} \otimes \mathbf{I}_K$$
contains the constraints relating the time shifts and frequency shifts, and $\mathbf{P}(M)$ is a block diagonal permutation matrix.  The vector $\mathbf{\beta}^*$ contains the true time shifts and frequency shifts of the targets (whose size is not a function of $M$).


\begin{theorem}(Asymptotically Perfect Recovery with Noise)\label{thm:asymptotic-noisy-recovery}
  Fix $\tau_{\max}$ and $f_{\max}$.  Choose $\{f_c^{a}\}$ for $a = 0,1,\ldots, M$ such that
  \begin{enumerate}
    \item $f_c^{a} \neq f_c^{b}$ for $a\neq b$, and
    \item $f_c^{a} = -f_c^{a-1}$ for $a$ odd.
  \end{enumerate}
  If $\mathbf{\gamma}(M)$ is a vector of zero mean independent random variables with finite variance, the solution $\mathbf{\hat{\beta}}(M)$ to
\begin{equation}\label{eq:frequency-linear-system-noisy}
  \mathbf{\hat{\beta}}(M) = \arg\min_{\mathbf{\beta}, \mathbf{P}(M)} ||\mathbf{A}(M)\beta - \mathbf{P}(M)\mathbf{\hat{\nu}}(M)||^2_2
\end{equation}
converges in probability to the true parameters as $M\to\infty$:
\begin{equation*}
  \mathbf{\hat{\beta}}(M) \overset{p}{\to} \mathbf{\beta}^*.
\end{equation*}

\end{theorem}

The proof relies on the weak law of large numbers and is provided in Appendix~\ref{app:proof-asymp-perfect-recovery-noise}.  Note that requirement (2) on $\{f_c^a\}$ can be relaxed at the expense of slower convergence of $\mathbf{\beta}(M)$.

\begin{algorithm}
  \caption{Recovering the time and frequency shifts from the noisy frequency and phase estimates.}
  \label{alg:noisy-delay-Doppler-recovery}
  \begin{algorithmic}[1]
    \STATE Data: $\hat{\nu}_k^m$ for  $k = 1,\ldots,K$ and $m=1,\ldots,M$
    \STATE Find possible target time and frequency shifts by solving
    $$\min_{\beta, \mathbf{P}} ||\mathbf{A}^m\mathbf{\beta} - \mathbf{P}^m\mathbf{\hat{\nu}}^m||^2_2$$
    for $m=1,\ldots,M$
    \FOR{$k=1,...,K$}
      \FOR{$\ell=1,...,K$}
        \STATE Calculate the slope $f_c(k,\ell)$ between possible target $k$ and possible target $\ell$ (as calculated in step 2)
      \ENDFOR
    \ENDFOR
    \STATE Find the slope that maximizes the distance from each of the slopes $f_c(k,\ell)$.  Send the $(M+1)^{th}$ pulse with this slope.
    \STATE The estimated time and frequency shifts are in $\mathbf{\beta}$
    \STATE Estimate the target amplitudes $\hat{c}_k$ using the least-squares estimate from Algorithm~\ref{alg:noisy-frequency-phase-recovery}
  \end{algorithmic}
\end{algorithm}


%

\section{Numerical Experiments}\label{sec:numerical-experiements}
We examine how well this procedure works using numerical experiments.  The KT algorithm provides perfect recovery of frequencies in the absence of noise, so we first generate a scene with targets at various time and frequency shifts.  We ensure that some of these targets are very close to each other in time and frequency.  The algorithm recovers all the targets to machine precision when two pulses of different chirp rate are used to generate the measurements.  We then explore the recovery of targets from noise-corrupted measurements by first examining the accuracy of the recovery procedure through Monte Carlo trials at a range of SNR environments.  The accuracy experiences a threshold SNR at which the error blows up.  We then pre-process the measurements with the denoising procedure to show that at the same input SNR range, the recovery has not yet reached the threshold.

\subsection{Noise-free Recovery}
We first ensure that we can perfectly recover the time shifts, frequency shifts, and amplitudes from uncorrupted samples.  The results are shown in Fig. \ref{fg:target-recovery-noiseless}.  The red \textrm{x} shows the true target location, while the blue circle shows the recovered parameters.  The recovery is to machine precision, meaning the errors were no larger than $10^{-10}$.

\begin{figure}[tp]
    \centering
    \includegraphics[width=0.99\columnwidth]{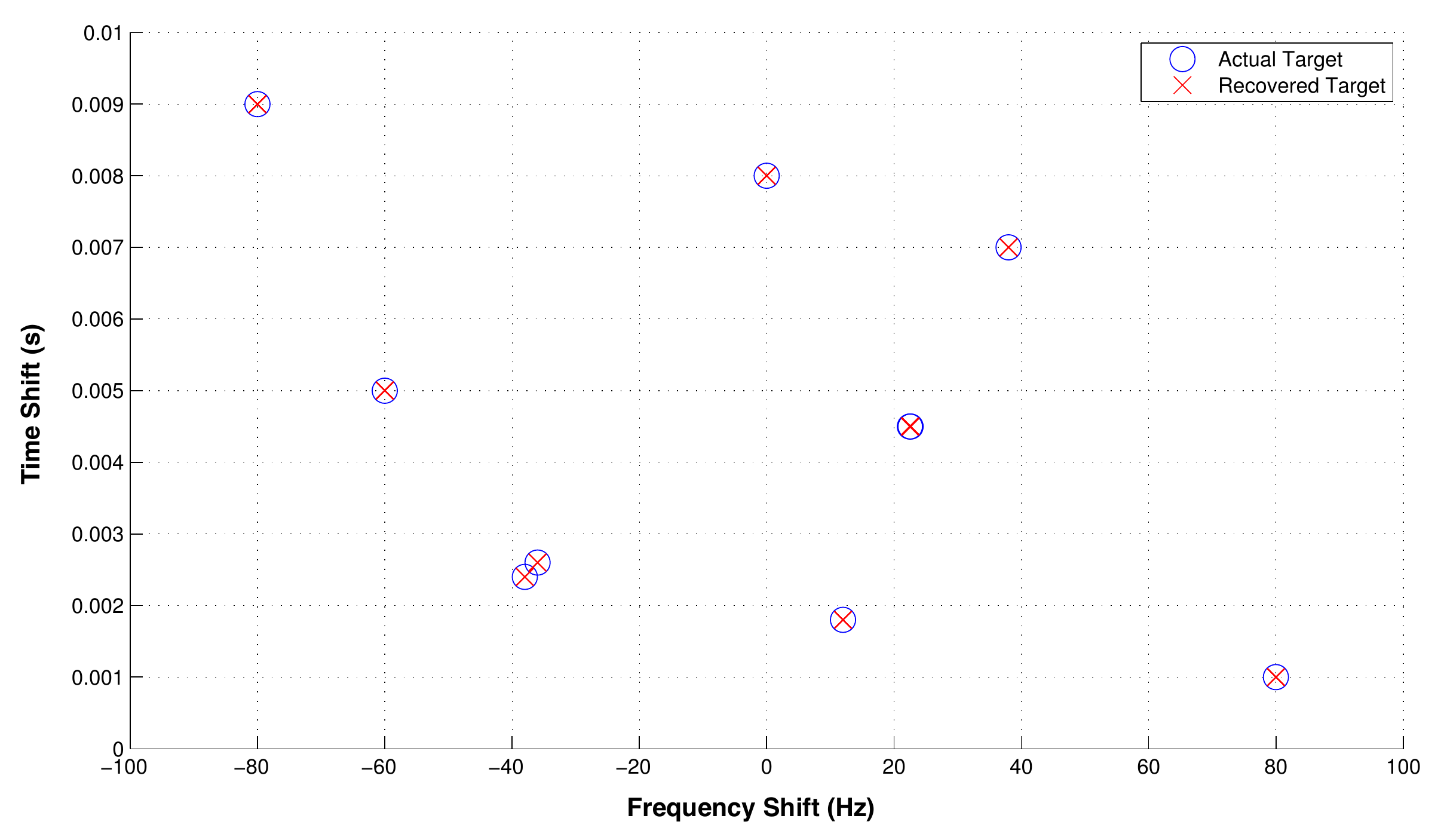}
  \caption{Recovery from noise-free measurements is perfect to machine precision.}
  \label{fg:target-recovery-noiseless}
\end{figure}

\subsection{Estimation of the Time and Frequency Shift from Noisy Samples}
The accuracy of the recovery procedure from noisy measurements, at various SNR environments, is shown in Fig. \ref{fg:recovery-error-noisy} for each parameter (time shift, frequency shift, and amplitude).  The RMSE at each SNR value is averaged over 1000 Monte Carlo trials in which a random realization of noise has been added to measurements containing $k=3$ targets with randomly generated parameters.  The time shift parameter was chosen uniformly from the interval $[0,\tau_{\max})$ and the frequency shift parameter uniformly from the interval $(-f_{\max}, f_{\max})$.  The error is proportional to the noise level up to about 10 dB, below which the errors start to grow much faster.  

To show the improvement that can be gleaned from denoising, Fig.~\ref{fg:recovery-error-noisy-denoised} shows the the results of running the same Monte Carlo trials described above, with the addition of the atomic norm denoising procedure described in Section~\ref{sec:denoising}.  The denoising is performed on the measurements after the random noise of the indicated \emph{input SNR} has been added.  This means that for a given SNR value in both Fig.~\ref{fg:recovery-error-noisy} and Fig.~\ref{fg:recovery-error-noisy-denoised}, AWGN of the same variance is added to the measurements.  The results in Fig.~\ref{fg:recovery-error-noisy-denoised} show that the threshold point is not reached even at $0$~dB SNR while it is reached at approximately $10$~dB SNR in Fig.~\ref{fg:recovery-error-noisy}.

\begin{figure}[tp]
  \centering
  \begin{subfigure}{0.33\columnwidth}
    \includegraphics[width=0.99\columnwidth]{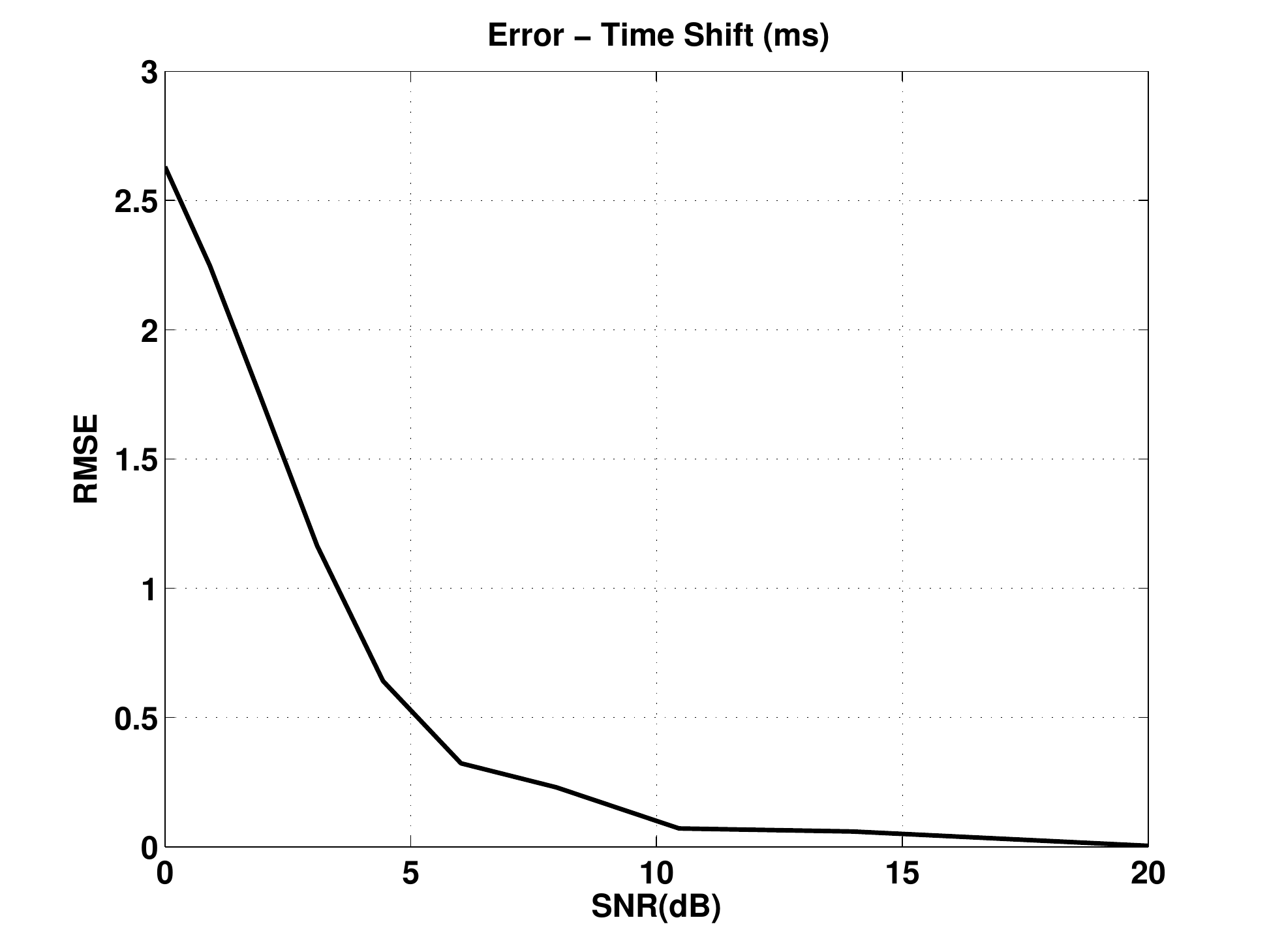}
    \caption{Time shift error}
    \label{fg:recovery-error-delay}
  \end{subfigure}
  \begin{subfigure}{0.33\columnwidth}
    \includegraphics[width=0.99\columnwidth]{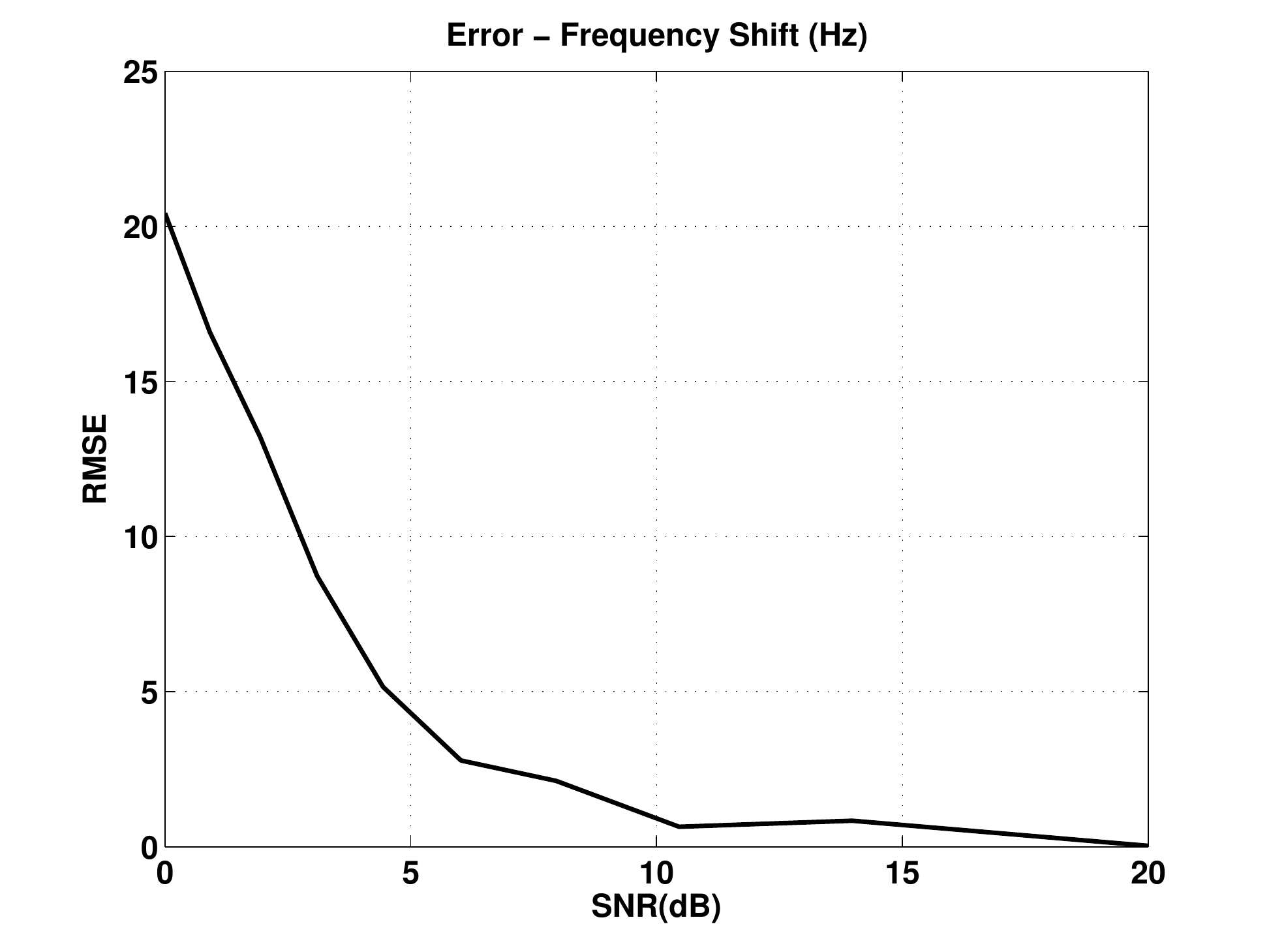}
    \caption{Frequency shift error} 
    \label{fg:recovery-error-Doppler}
  \end{subfigure}
  \begin{subfigure}{0.32\columnwidth}
    \includegraphics[width=0.99\columnwidth]{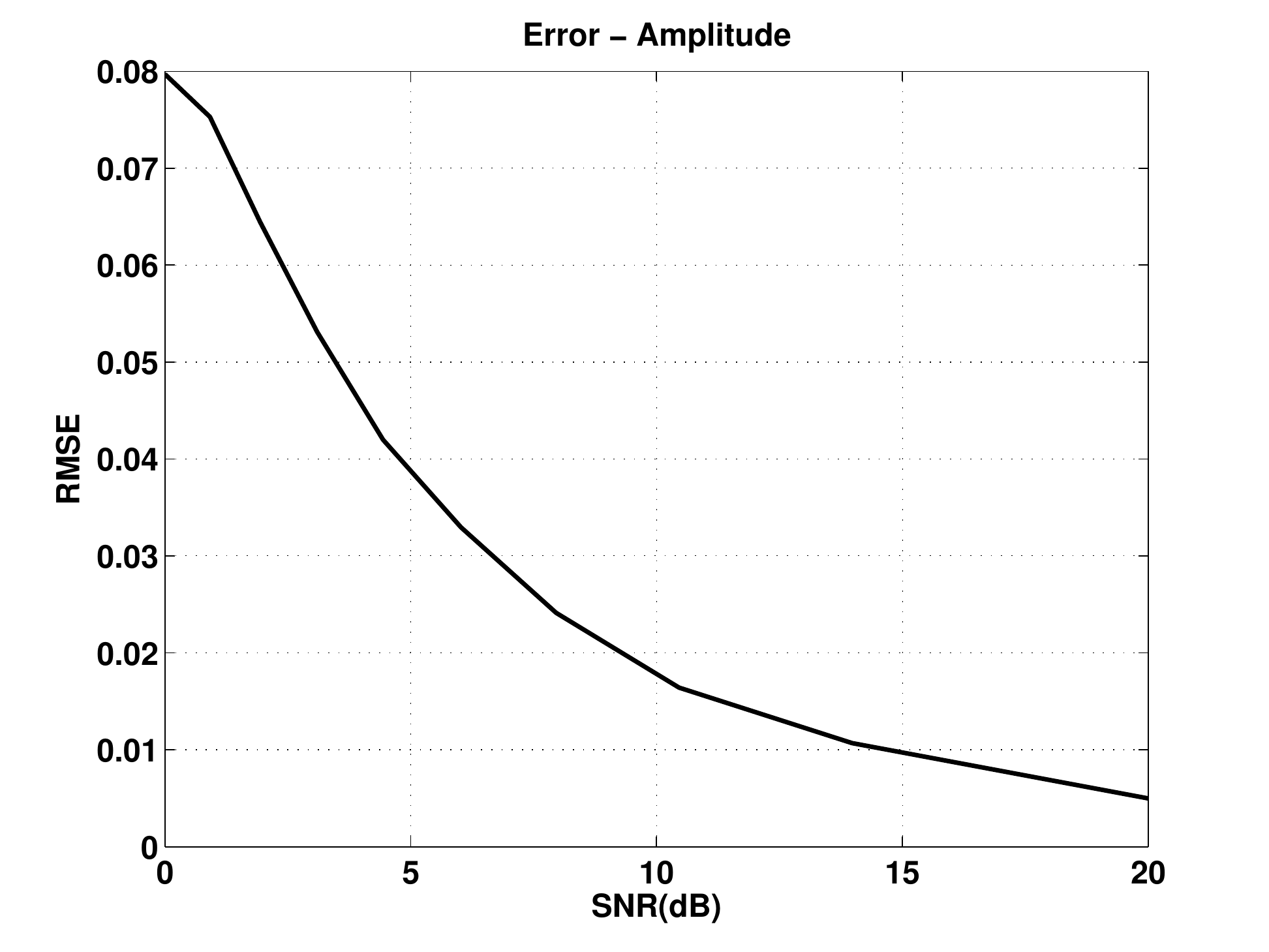}
    \caption{Amplitude error}
    \label{fg:recovery-error-amplitude}
  \end{subfigure}
  \caption{Recovery error from noisy measurements over 1000 Monte Carlo trials.  Note the threshold at $\approx 10$ dB below which the error increases rapidly.}
  \label{fg:recovery-error-noisy}
  \vspace{-0.1in}
\end{figure}

\begin{figure}[tp]
  \centering
  \begin{subfigure}{0.33\columnwidth}
    \includegraphics[width=0.99\columnwidth]{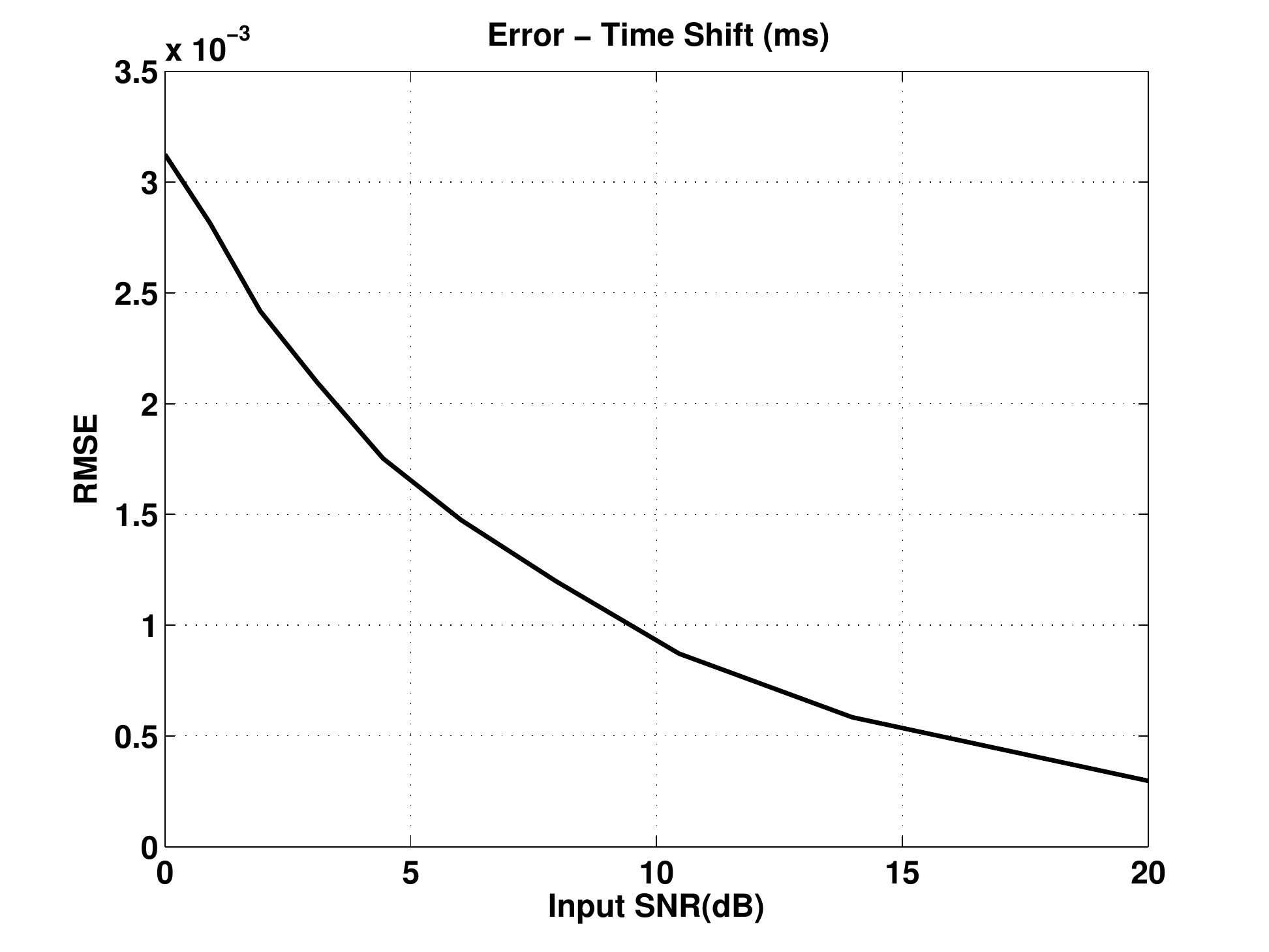}
    \caption{Time shift error}
    \label{fg:recovery-error-delay-denoised}
  \end{subfigure}
  \begin{subfigure}{0.33\columnwidth}
    \includegraphics[width=0.99\columnwidth]{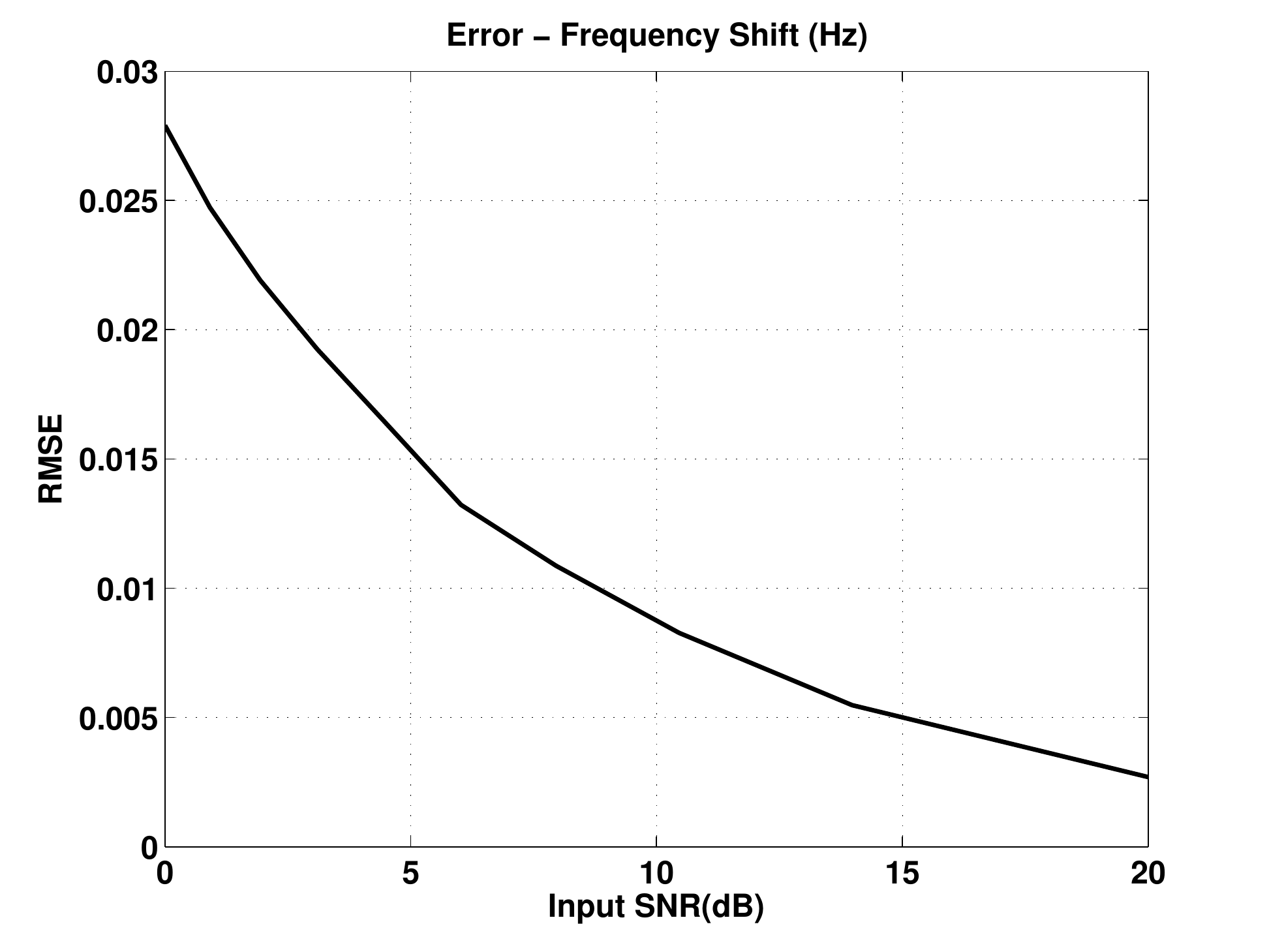}
    \caption{Frequency shift error} 
    \label{fg:recovery-error-Doppler-denoised}
  \end{subfigure}
  \begin{subfigure}{0.32\columnwidth}
    \includegraphics[width=0.99\columnwidth]{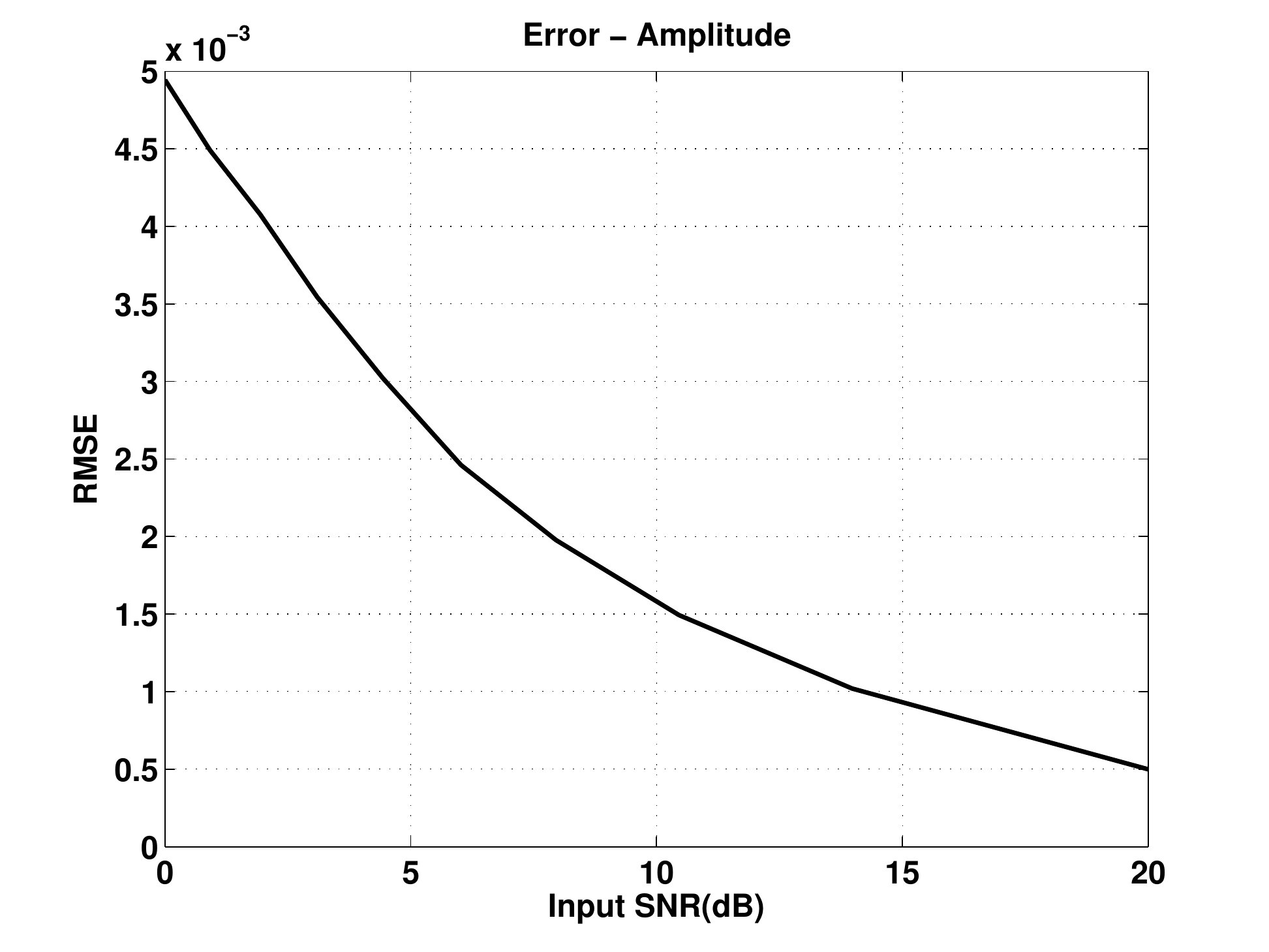}
    \caption{Amplitude error}
    \label{fg:recovery-error-amplitude-denoised}
  \end{subfigure}
  \caption{Recovery from noisy measurements that are first denoised using atomic norm denoising.}
  \label{fg:recovery-error-noisy-denoised}
  \vspace{-0.2in}
\end{figure}

Finally, to investigate the effective resolution limits of the approach, in Fig. \ref{fg:recovery-error-noisy-limited} we performed the same Monte Carlo trials described above but the time shift parameter has been drawn uniformly at random from the interval $[0,\tau_{\max}/10)$ and the frequency shift parameter has been drawn uniformly at random from the interval $(-f_{\max}/10, f_{\max}/10)$.  This ensures that more targets are chosen with parameters that are closely spaced in the time shift frequency shift plane.  The error is shown to still be roughly proportional to the amount of noise down to about 20 dB SNR, at which point it grows faster.  This seems to be consistent with Fig. \ref{fg:recovery-error-noisy} as the threshold, at which point the error grows rapidly, occurs at a 10 dB higher SNR because the reduction in the parameter space, from which the parameters are chosen, has been reduced by a factor of 10 in each dimension.  This means that the parameters are likely to be much closer to each other.  We emphasize here that the limits on resolution using this method is an open problem and requires further investigation and that the choice of LFM parameters, namely the chirp rate $f_c^m$, has not been optimized in any way for recovery.  These average-case errors seem to align with our intuition on the recovery.

\begin{figure}[tp]
  \centering
  \begin{subfigure}{0.33\columnwidth}
    \includegraphics[width=0.99\columnwidth]{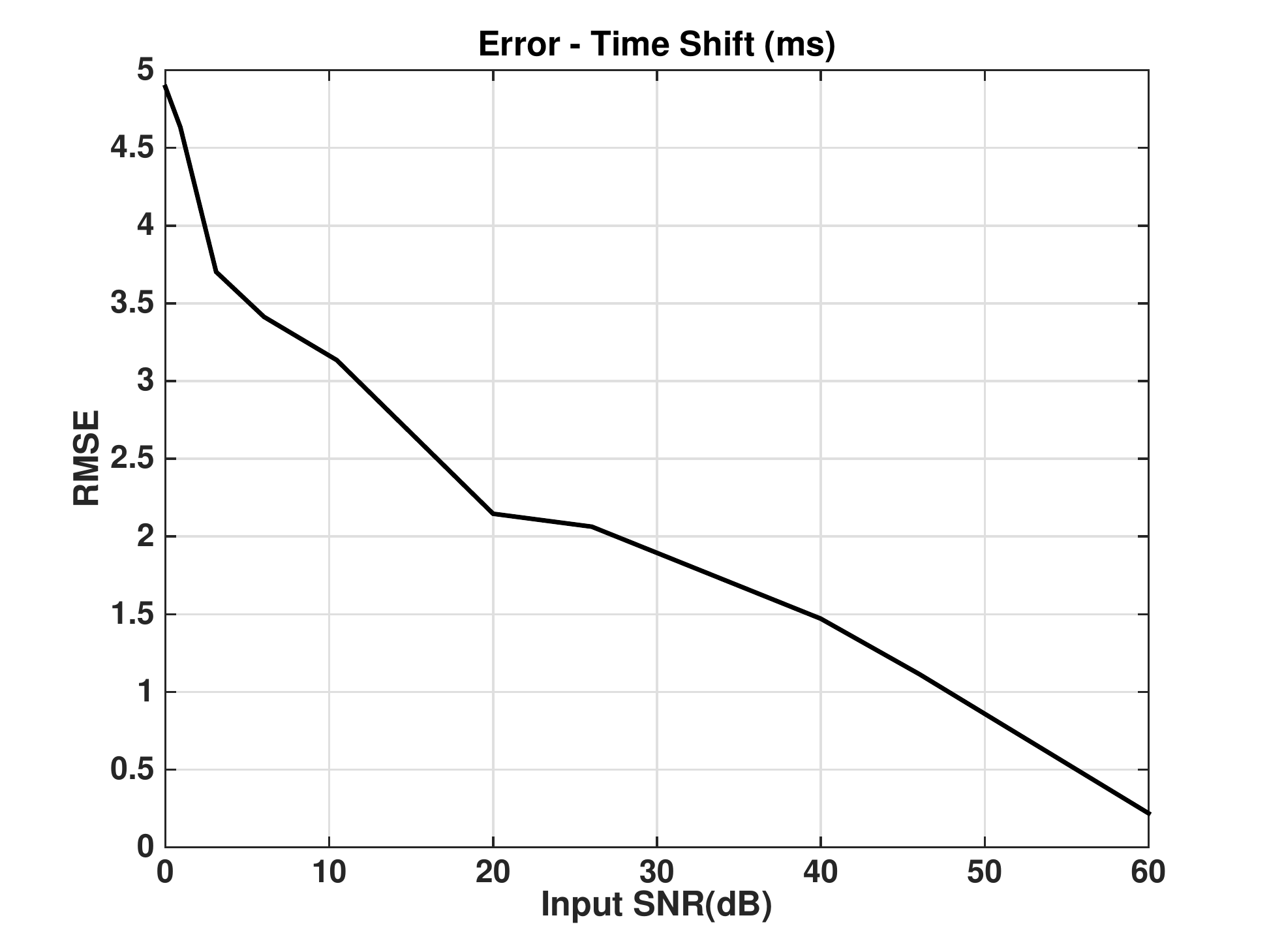}
    \caption{Time shift error}
    \label{fg:recovery-error-delay-limited}
  \end{subfigure}
  \begin{subfigure}{0.33\columnwidth}
    \includegraphics[width=0.99\columnwidth]{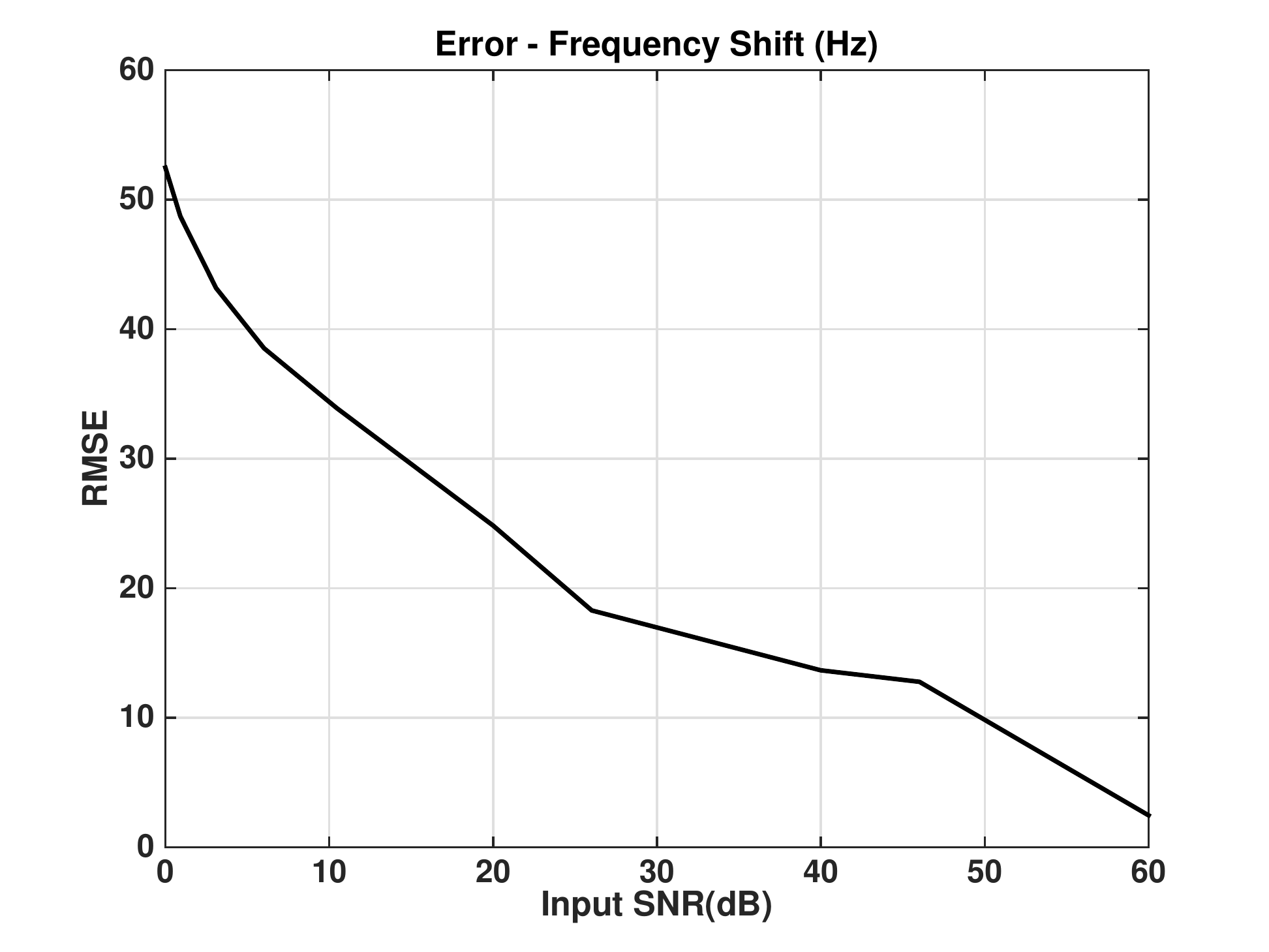}
    \caption{Frequency shift error} 
    \label{fg:recovery-error-Doppler-limited}
  \end{subfigure}
  \begin{subfigure}{0.32\columnwidth}
    \includegraphics[width=0.99\columnwidth]{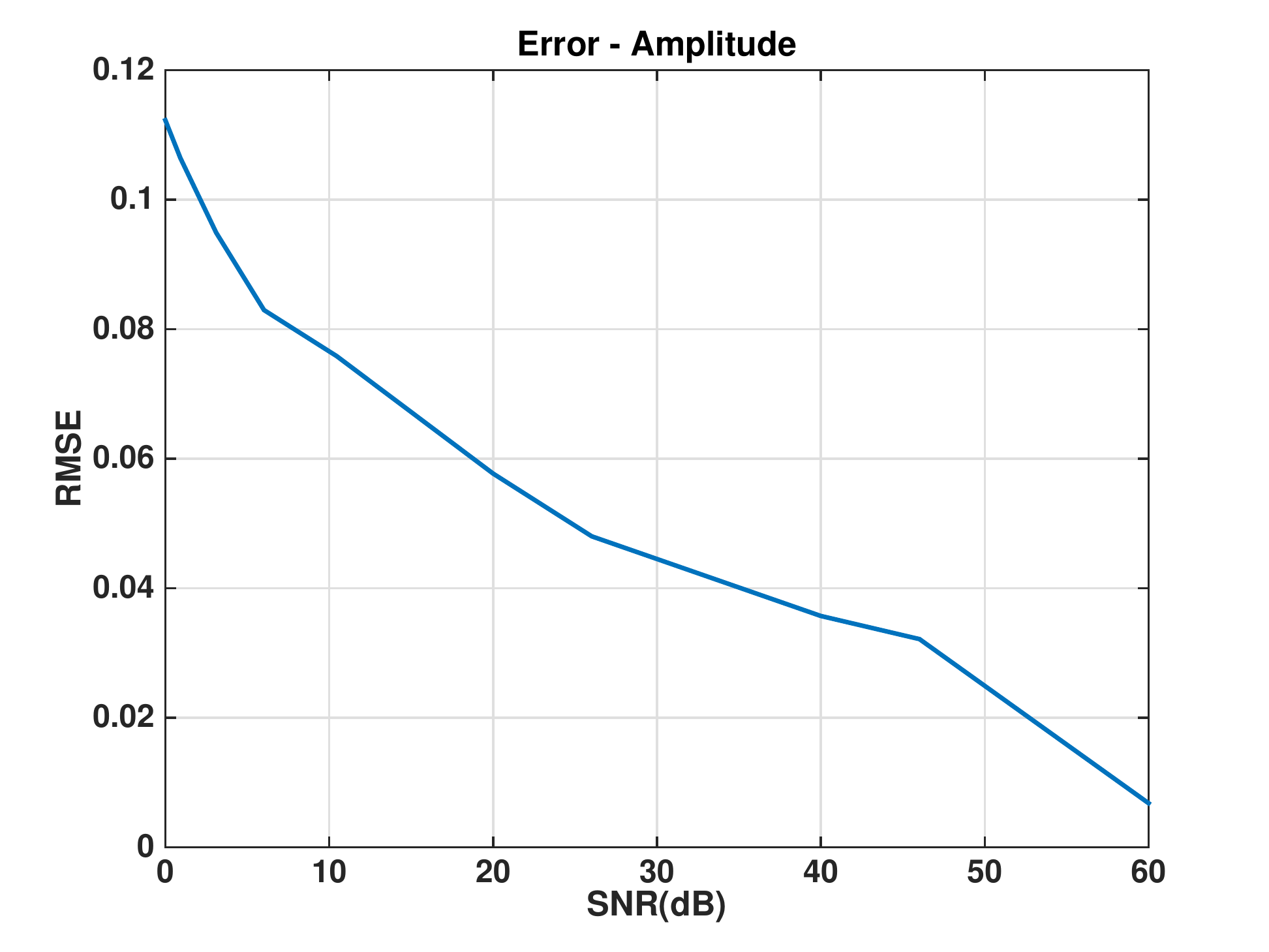}
    \caption{Amplitude error}
    \label{fg:recovery-error-amplitude-limited}
  \end{subfigure}
  \caption{Recovery error from noisy measurements over 1000 Monte Carlo trials with the parameters limited to $[0,\tau_{\max}/10)$ and $(-f_{\max}/10, f_{\max}/10)$.}
  \label{fg:recovery-error-noisy-limited}
  \vspace{-0.1in}
\end{figure}

\section{Conclusions and Future Work}\label{sec:conclusion}
We have proposed a novel technique for identifying LTV systems using LFM pulses as probing waveforms.  This technique leverages a vast array of frequency recovery, or estimation, algorithms in the existing literature.  We have shown that the KT algorithm, along with a denoising procedure, provides excellent numerical results in simulated Monte Carlo trials.  The primary advantage to our approach is that the resources needed (e.g., bandwidth and acquisition time) scale proportionally to the complexity of the LTV system, i.e., the number of scatterers in a radar scene or multi path sources in a communication system.  We also operate in a continuous parameter space, so we have resolution limits that scale with the amount of noise present in the measurements, all the way down to infinite resolution if there is no noise.  A full analysis of the resolution limits using this approach is an open problem that requires non-asymptotic versions of the main results.  Additionally, the optimal choice of LFM chirp rates $f_c^m$ is an open problem.

\appendices

\section{Proof of Theorem \ref{thm:asymptotic-noisy-recovery}}\label{app:proof-asymp-perfect-recovery-noise}
Two sources of errors can cause errors in the solution to \eqref{eq:frequency-linear-system-noisy}, and the proof relies on proving that each has vanishing influence asymptotically.  First, the corruption of the frequency estimates $\hat{\nu}_k = \nu_k + \gamma_k$ by $\gamma_k$ causes an error in the estimate of the target parameters.  The law of large numbers ensures that this source of error has vanishing influence.  Second, the unordered nature of the estimated frequencies, captured by the permutation matrix $\mathbf{P}$, means that multiple solutions can exist for finite $M$.  Diversity in the selection of $f_c^m$ ensures that the solution to \eqref{eq:frequency-linear-system-noisy}, simultaneously for all $M$ pulses, is unique asymptotically.  We note here that this latter source of error is highly dependent on the particular LTV system (or target scene) and the choice of $f_c^m$.  The $f_c^m$ can be adaptively chosen to most effectively identify the LTV parameters.

With the mild constraint that $f_c^m \neq f_c^q$ for some $m\neq q$, i.e., every pulse does not have the same chirp rate, we can use the Moore-Penrose pseudoinverse of $\mathbf{A}$ to write \eqref{eq:frequency-linear-system-noisy} as
\begin{equation}\label{eq:frequency-linear-system-noisy-inverse}
  \min_{\mathbf{\beta}, \mathbf{P}} ||\mathbf{\beta} - (\mathbf{A}^*\mathbf{A})^{-1}\mathbf{A}^*\mathbf{P}\mathbf{\hat{\nu}}||^2_2.
\end{equation}
The estimated frequencies $\hat{\mathbf{\nu}} = \mathbf{\nu} + \mathbf{\gamma}$ we assume are corrupted by $\mathbf{\gamma}$ that is independent and has finite variance.  The term $(\mathbf{A}^*\mathbf{A})^{-1}\mathbf{A}^*\mathbf{P}\mathbf{\gamma}$ is the error in the parameter estimates due to the error in the frequency estimates (which is in turn due to the noise).  We will show that this term gets smaller as more pulses are processed.  Recall that $\mathbf{A} = \mathbf{B}\otimes \mathbf{I}_K$, and we can write the first term in the pseudoinverse $(\mathbf{A}^*\mathbf{A})^{-1} = ((\mathbf{B}^*\mathbf{B})^{-1}) \otimes \mathbf{I}_K$.  The matrix $\mathbf{B}$ contains two columns, the first is all ones and the second contains the chirp rates of the pulses $f_c^m$.  If we restrict our choice of chirp rates such that for $m$ odd, the chirp rate is $f_c^m = -f_c^{m-1}$, then the two columns of $\mathbf{B}$ have zero inner product and the Gram matrix is diagonal\footnote{Note that this requirement can be relaxed, but the resulting Gram matrix is not diagonal and the estimate variance vanishes more slowly.}.  The pseudoinverse of $\mathbf{A}$ is thus $((\mathbf{B}^*\mathbf{B})^{-1}\mathbf{B}^*) \otimes \mathbf{I}$.  Let us write $\mathbf{G} = \mathbf{B}^*\mathbf{B}$ where $\mathbf{G}$ is diagonal with $g_{11} = M$ and $g_{22} = \sum_{m=1}^M 4(f_c^m)^2 \geq 4M\min_{m}(f_c^m)^2$.  The pseudoinverse of $\mathbf{B}$ is thus a scaled version of $\mathbf{B}^*$.  The first column of $\mathbf{B}$ is scaled by $1/M$ and the second columns is scaled by $1/\sum_{m=1}^M 4(f_c^m)^2$, which is bounded by a quantity proportional to $1/M$.  The error term for each target parameter, after carrying out the matrix multiplication, is a scaled sum of each individual frequency error.  Concentrating on the frequency shifts, which are contained in the first half of the vector $\beta$, we get the error for the $k^{th}$ frequency shift to be
$$\sum_{m=1}^M\frac{1}{M}\gamma_k^m.$$
Each $\gamma_k^m$ has the same variance $\tilde{\sigma}^2$, which approaches \eqref{eq:CRB-freq}, and these are all independent of each other.  The independence is a direct result of the assumption of independent noise.  Each pulse is processed separately by the same deterministic recovery algorithm, so if the noise in the input to the algorithm is independent, then the errors in the output of the algorithm are also independent.  The variance of the scaled sum is therefore reduced by the factor $1/M$, so the error is vanishing as $M\to\infty$.

The matrix $\mathbf{P}$ affects the recovery for a small number of pulses by introducing ambiguity into the recovered parameters, i.e., even in the absence of noise, the solution to \eqref{eq:frequency-linear-system-noisy-inverse} might not be unique.  By requiring a variety of chirp rates in the pulses, we ensure the set of solutions from each pulses cluster around the true parameters while the ambiguous solutions are spread in the time shift--frequency shift plane.  Asymptotically, the clustering becomes tighter because the noise variance is vanishing with $1/M$.

\section{Ambiguous Phase Terms}\label{sec:ambiguous-phase-terms}
Recall the delay-phase mapping \eqref{eq:phase-mapping-positive}.  To highlight the dependence on the delay $\tau_k$, let us write $\theta^{m}(\tau)$.  To prevent ambiguous phase terms, the complex exponential $\mathrm{e}^{j2\pi\theta^m(\tau)}$ must be bijective over $0\leq \tau \leq \tau_{\max}$, which requires the quadratic function $\theta^{m}(\tau)$ to be bijective over $0\leq \tau \leq \tau_{\max}$ and its range limited to an interval of length at most 1.  The following lemma provides necessary and sufficient conditions to prevent these ambiguities.

\begin{lemma}[Ambiguous Phase Terms] \label{lemma:ambiguous-phase}
Let $0 \leq \tau \leq \tau_{\max}$.  The function (of $\tau$)
$\mathrm{e}^{j2\pi\theta^{m}(\tau)} = \mathrm{e}^{j2\pi f_c^{m}\tau^2}$
is bijective if and only if
\begin{equation}\label{eq:fc-constraint}
  0 < |f_c^{m}| \leq \frac{1}{\tau_{\max}^2}.
\end{equation}

\end{lemma}

\begin{proof}
A function is bijective over an interval if and only if it is strictly increasing or decreasing over that interval.  We examine $\theta^{m}(\tau)$ and its derivative on the interval $0 \leq \tau \leq \tau_{\max}$.  The derivative is
\begin{equation}\label{eq:psi-derivative}
  \frac{d}{d\tau}\theta^{m}(\tau) = 2f_c^{m}\tau.
\end{equation}

At $\tau = 0$, $\theta^{m}(0) = 0$ and $\frac{d}{d\tau}\theta^{m}(\tau)|_{\tau = 0} = 0$.  At $\tau = \tau_{\max}$, $\theta^{m}(\tau_{\max}) = f_c^m\tau_{\max}^2$ and the derivative also has the same sign as $f_c^m$.  In either case, $\theta^{m}(\tau)$ is monotonically increasing or decreasing on $0\leq \tau \leq\tau_{\max}$.
We require
\begin{equation}\label{eq:psi-constraint}
  |\theta^{m}(\tau_{\max})| = |f_c^{m}|\tau_{\max}^2 < 1
\end{equation}
leading directly to \eqref{eq:psi-derivative}.
\end{proof}

 \bibliographystyle{IEEEtran}
 \bibliography{IEEEabrv,andrew.bib}

\end{document}